\documentclass[journal,twoside,web]{ieeecolor}
\usepackage{generic}

\usepackage{mathrsfs}
\usepackage{color}
\usepackage{amsfonts}
\usepackage{subfigure}
\usepackage{booktabs}
\usepackage{cite}
\usepackage{amsmath,amssymb,amsfonts}
\usepackage{algorithmic}
\usepackage{graphicx}

\usepackage{enumerate}

\usepackage[ruled]{algorithm2e} 
\usepackage{algorithmic}

\usepackage{bm}
\usepackage{makecell}
\usepackage{balance}
\usepackage{mathtools}
\usepackage{bbm}
\usepackage{dsfont}
\usepackage{pifont}
\usepackage{cite}

\usepackage{threeparttable}
\usepackage{textcomp}

\newtheorem{remark}{Remark}
\newtheorem{theorem}{Theorem}
\newtheorem{lemma}{Lemma}

\newtheorem{assumption}{Assumption}

\begin{document}
\title{System Identification with Variance Minimization via Input Design}
\author{Xiangyu Mao$^\dag$,\IEEEmembership{Member, IEEE}, Jianping He$^\dag$,\IEEEmembership{Member, IEEE}, and Chengcheng Zhao$^\ddag$ ,\IEEEmembership{Member, IEEE}
	\thanks{
	 $^\dag$: The Dept. of Automation, Shanghai Jiao Tong University, and Key Laboratory of System Control and Information Processing, Ministry of Education of China, Shanghai, China. E-mail address: \{maoxy20, jphe\}@sjtu.edu.cn. 
	 $^\ddag$: The State Key Laboratory of Industrial Control Technology and Institute of Cyberspace Research, Zhejiang University, China. E-mail: zccsq90@gmail.com.
     Preliminary results have been submitted to the 2022 American Control Conference \cite{previ}.
	}%
}

\maketitle

\begin{abstract}
The subspace method is one of the mainstream system identification method of linear systems, and its basic idea is to estimate the system parameter matrices by projecting them into a subspace related to input and output. However, most of the existing subspace methods cannot have the statistic performance guaranteed since the lack of closed-form expression of the estimation.
Meanwhile, traditional subspace methods cannot deal with the uncertainty of the noise, and thus stable identification results cannot be obtained.
In this paper, we propose a novel improved subspace method from the perspective of input design, which guarantees the consistent and stable identification results with the minimum variance. Specifically, we first obtain a closed-form estimation of the system matrix, then analyze the statistic performance by deriving the maximum identification deviation.
This identification deviation maximization problem is non-convex, and is solved by splitting it into two sub-problems with the optimal solution guaranteed.
Next, an input design method is proposed to deal with the uncertainty and obtain stable identification results by minimizing the variance. This problem is formulated as a constrained min-max optimization problem. The optimal solution is obtained from transforming the cost function into a convex function while ensuring the safety constraints through the method of predictive control.
We prove the consistency and the convergence of the proposed method.
Simulation demonstrates the effectiveness of our method.
\end{abstract}


\IEEEpeerreviewmaketitle

\section{Introduction}
System identification refers to determining a mathematical model to describe system behavior according to the observed data and prior system knowledge \cite{040}.
The process of system identification can be seen as a systematic procedure that builds mathematical models from observation, which plays the vital role of the interface between the real world of applications and the mathematical world of control theory and model abstractions \cite{LJUNG20101}.
System identification helps to describe the system dynamical properties and brings fundamental knowledge for further research such as prediction and control.
The identification model can be used to perform parameter correction and real-time adaptive control of the high-precision system \cite{7394584}.
According to the predicted model quantities by filtering and identifying acquired measurements, the attack or abduction of robot network can be realized \cite{8951600,7057677}.
Besides, system identification is widely used in large-scale network systems and nonlinear systems \cite{8897147} because it is usually not feasible to obtain the large-scale network model directly, especially when the network is distributed \cite{8357807}.

In this paper, we consider the identification of linear systems, which is a parameter estimation problem of gray-box models \cite{ASTROM1971123}.
In the literature, the mainstream methods for system identification fall into two classes, the prediction error method (PEM) and the subspace method \cite{030}.
The PEM estimates the system matrices to predict the output by input and historical output data and aims to minimize a cost function related to the prediction error \cite{032}.
The PEM has the best possible asymptotic accuracy \cite{002} and the predictions are useful in model-based control methodologies such as predictive control.
For subspace method, it is to identify a similarity transformation of the system by projecting the system parameter matrices into a subspace related to input and output \cite{036}. 
The subspace method has the advantages of noniterative solutions for the state space form of the systems.
These solutions are convenient for estimation, filtering, prediction and control \cite{034}.

\subsection{Motivation}
Despite the prominent contributions of the pioneering works, there still remain some notable issues.
First, the traditional system identification methods are with high computational complexity, which makes them hard to be used for real-time applications.
In detail, the PEM often leads to a non-convex, multi-dimensional and nonlinear optimization problem \cite{031}, and the numerical, iterative procedures for solving such problems are guaranteed only to find local minima \cite{034}. 
Another traditional identification method, the subspace method has more reliable noniterative numerical solutions.
However, the linear algebraic steps such as the matrix factorization applied in the subspace method do not provide a cost function like the PEM, making the statistic analysis of the subspace method much more difficult. 
Although instrumental variable is a feasible tool to simplify the subspace method and the error analysis \cite{004}, the use of instrumental variables reduces the accuracy \cite{021}.
Second, traditional identification methods are difficult to effectively deal with the uncertainty of the system to obtain stable and robust identification results \cite{007}.
This phenomenon is more obvious when the amount of data is small, where both the prediction by the PEM and the identification result of the subspace method tend to have a large variance.
Third, the input signal of the system needs to be designed, which plays a vital role in the process of identification because it directly affects the output and the identification results.
However, 
for traditional identification methods, the excitation signal of the system is mainly white noise \cite{002,004}, which may result in unstable identification results of system matrices with a large variance.
Optimal input design for the subspace method remains an open issue \cite{034}.

\begin{figure*}[t]
  \centering 
  \setlength{\abovecaptionskip}{0.1cm}
    \label{sae0}
    \includegraphics[width=0.75\textwidth]{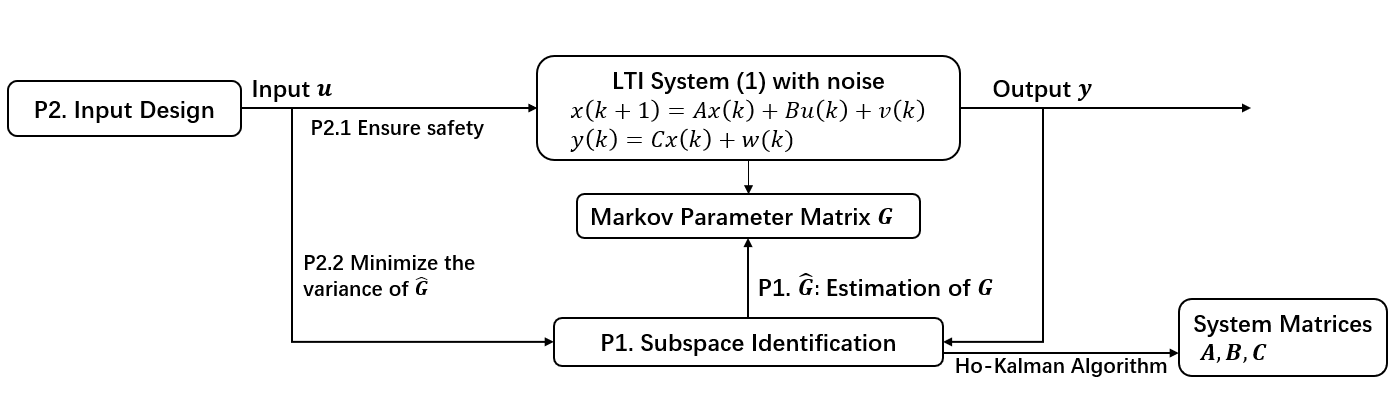} 
  \caption{Framework of this paper.} 
  \vspace*{-10pt}
\end{figure*}

The above issues have motivated the study of this paper.
We aim to propose an improved subspace method and deal with the uncertainty by designing the input to minimize the identification variance. 
Different from traditional subspace methods which derive an implicit solution to the system matrices through linear algebraic steps, our method has a closed-form consistent solution.
Under the condition that the system input signal can be designed, we obtain the optimal input to achieve stable identification results. 
Furthermore, the proposed method does not require strong assumptions including system stability or zero initial states \cite{029}, which narrow the scope of the application of the identification method.

\subsection{Contribution}
To start with, we analyze a linear time-invariant (LTI) system and propose a subspace method to identify the Markov parameter matrix.
The Markov parameter matrix is obtained by a closed-form expression related to the input and output. 
Next, we analyze the statistic property and obtain the maximum identification deviation of the proposed subspace method.
Finally, the optimal input signal for minimizing the identification variance is obtained to account for uncertainties associated to the noise.
We propose an input design algorithm based on the statistic property analysis to minimize the maximum identification deviation, which is provably equivalent to minimize the variance. 
We analyzed the performance of the algorithm.
It is proved that the proposed input design method has a faster convergence rate of identification error. 
More importantly, stability of identification can be achieved by the proposed method.
It shows that the maximum deviation of identification results of the proposed method converges at a speed of $\mathcal{O}(N^{-1})$, while using white noise input can only converge in probability at the speed of $\mathcal{O}(N^{-\frac{1}{2}})$.

The differences between this paper and its conference version \cite{previ} include i) optimal input design problem is solved rigorously with detailed procedure and analytical solutions provided, ii) the performance analysis of the proposed algorithm are provided, especially the error analysis and the deviation analysis, iii) extended simulations are provided.

The main contributions are summarized as follows. 
\begin{itemize}
\item We propose an improved subspace system identification method with a closed-form, consistent estimation of the system matrix which helps to tackle the difficulty of statistic analysis of subspace identification.
The method avoids the usage of instrumental variables or the requirement of system stability or zero initial states. 
Then, the maximum identification deviation and the upper bound of the identification error are provided.
                                  
\item  An input design algorithm is proposed to deal with the uncertainty in system identification while ensuring the safety constraints.
The proposed algorithm enables the observer to obtain stable identification results, which have the minimum identification variance.
Simulations demonstrate that compared with the commonly-used white noise input, the accuracy and stability of identification under our algorithm have a considerable improvement. 

\item We solve the maximum identification deviation problem and input design problem, and provide the optimal solution of these two problems, respectively.
Both of them are non-convex optimization problems. 
The problem is solved by splitting it into two quadratic programming sub-problems and relaxation with the optimal guaranteed. 
The input design problem is formulated as a constrained min-max problem.
We transformed the cost function into a convex function and obtain the optimal solution by gradient descent while ensuring the safety constraints through predictive control.
 
\end{itemize}

This paper provides deeper insights into the system identification and the input design of system. 
The theoretical results can serve as instructions to design the optimal control  for the sake of accuracy and stability of  identification, and also beckon further research to explore more advanced methods for general system models such as nonlinear systems.

The remainder of this paper is organized as follows. 
Section \ref{Relatedwork} provides literature research of system identification and input design.
Section \ref{preliminary} gives the notations and describes the problem of interest. 
Section \ref{subspace} proposes a subspace identification method and then analyzes the maximum identification deviation of the proposed subspace method. 
The input design algorithm to minimize the identification variance or the maximum deviation is given in Section \ref{inputdesignsection}. 
Simulation results are shown in Section \ref{Sim}, followed by conclusions and future directions in Section \ref{conclusionfi}.

\section{Related Work}\label{Relatedwork}
There have been extensive researches on system identification and input design in the literature. This section gives a brief overview of the PEM method, the subspace method, and input design in system identification.

\subsubsection{The Prediction Error Method}
The PEM aims to estimate the system matrices in a predictor $\hat{y}$ of output $y$ such that $\hat{y}$ approximates $y$ \cite{034}.
The PEM obtains minimum-variance estimates of the dynamics both of the deterministic and of the stochastic part of the system and has excellent statistic properties. 
Since the prediction-error cost function of the PEM is mostly non-convex and complicated,
it is necessary to simplify the objective function or constraints of the PEM.
Using the low-rank nature of the block Hankel matrix for decomposition \cite{001}, fitting the input signal by a polynomial \cite{006} or simplifying the original problem by L1 regularization and LASSO \cite{005} are efficient ways.
Another open issue is how to guarantee the optimal prediction in the simplified or relaxed PEM and it remains a relevant research topic.

\subsubsection{The Subspace Method}
The key idea of the subspace method is that certain subspaces related to the system matrices can be retrieved via linear-algebra steps such as RQ factorization, when storing the input and output data.
Compared with the PEM, the subspace method has more reliable noniterative numerical solutions because it does not need to solve a complex prediction-error minimization problem or parameterize the system model \cite{034}.
During the process of subspace identification, the system state can be estimated simultaneously \cite{013}, and the observer can continuously update the identification results. 
Thus, the subspace method has been widely applied in the identification of large-scale network systems \cite{003,016,012} or other complex models represented by the glucose-insulin model \cite{015}. 
To make the subspace method easier to implement,  in \cite{002}, kernel norm is used as the relaxation of the rank constraints which are commonly used. In \cite{017}, the Sylvester equation is adopted, combined with the traditional subspace method. 
Since the subspace method uses linear algebra steps such as RQ decomposition and SVD decomposition, which have no explicit expressions, it is difficult to obtain a closed-form solution of the identification result.
Consequently, lack of a closed-form solution of identification results makes it difficult for statistic analysis in subspace methods.

\subsubsection{Input Design in System Identification}
There have been research works which investigate the input design in the PEM \cite{025,024,041} based on optimization theories.
The input design problem of minimizing the maximum error of identification is discussed under a finite-impulse response (FIR) model based on the PEM \cite{023,027}. 
Input design for minimizing the variance in system identification mainly considers maximizing the Fisher information of the system, while this method is only applicable to the PEM \cite{026}. 
For more general methods in system identification, input design is considered for the output error \cite{STOJANOVIC2016576}, control \cite{MU2018327} or the maximum information\cite{FUJIMOTO201837}.
However, few researches focus on input design problem in the subspace method, especially for the stability of identification, because of the lack of the closed-form solution.

Different from the aforementioned works, as shown in Table 1, this paper focuses on subspace identification and input design for stability. We use the designed input, not the white-noise  to minimize the identification variance according to the closed-form identification, so that the identification result can remain stable when the system is with noise. We give the convergence analysis of the identification error and variance.

\begin{figure*}[htbp]
  \begin{center}
  \setlength{\abovecaptionskip}{0.1cm}
    \label{theorems0}
    \includegraphics[width=0.65\textwidth]{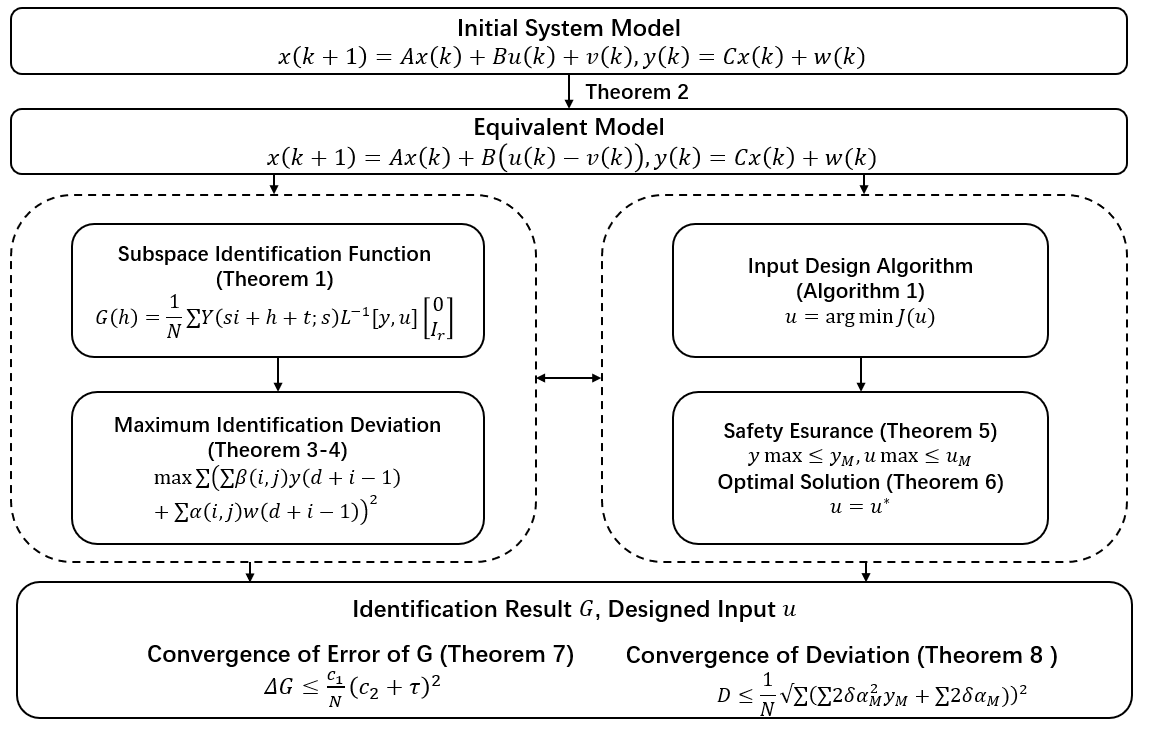} 
  \caption{Roadmap of the main theoretical results in this paper.} 
  \end{center}
  \vspace*{-10pt}
\end{figure*}

\section{Preliminaries and Problem Formulation}\label{preliminary}

\subsection{Basic Model and Problem Formulation}

We investigate the state-space model of a discrete-time LTI system, defined by
\begin{equation}
\begin{array}{ll}
\left\{\begin{array}{l}\label{sysmodel0}
x(k+1)=A x(k)+Bu(k)+v(k), \\
y(k)=C x(k)+w(k),
\end{array}\right.
\end{array}
\end{equation}
where $x(k)\in\mathbb{R}^m$ is the state variable, $y(k)\in\mathbb{R}^n$ is the output signal, $u(k)\in\mathbb{R}^p$ is the input signal, $v(k)\in \mathbb{R}^m, w(k)\in\mathbb{R}^n$ are the process noise and output noise, $A\in\mathbb{R}^{m \times m}, B\in\mathbb{R}^{m \times p}, C\in   \mathbb{R}^{n \times m}$ are system matrices.

The work of this paper is divided into two steps, subspace identification and input design.  
Subspace identification is to estimate the system matrices $A,B,C$ via the observation of system input $u$ and output $y$. 
Then, the input design mechanism is proposed to ensure the stability of identification and the safety of the system, i.e., designing the input $u$ to minimize the estimation variance of $A,B,C$ when there exists noise $v,w$ and ensure that $y$ is in safety range.

Define the extended Markov parameter matrix $G(t)$ by
\begin{equation}\nonumber
G(t) = \left[ CA^{t-1}B, \ \ CA^{t-2}B,\  \ \cdots,\ \  CAB, \ \ CB\right].
\end{equation}
It follows that $G$ can be directly derived by $A,B,C$.
On the other hand, the system matrices $A$, $B$, $C$ can be obtained up to a similarity transformation form from $G$ using the Ho-Kalman Algorithm.
Hence, subspace identification process can be transformed into the problem of identifying $G$.
Moreover, the subspace method only identifies a similar transformation of the system \cite{036} and $G$ is invariant under the similarity transformation. 
Therefore, we choose the Markov parameter matrix $G$ as the goal of subspace identification.

Therefore, for any LTI system, we identify the system via estimating the Markov parameter matrix $G$ by the input and output. 
The input is designed aimed at ensuring safety and minimizing the variance of estimation of $G$. The system matrices $A,B,C$ are derived by the Ho-Kalman Algorithm from $G$. The framework of this paper is shown in Fig.1. 

Both subspace identification and input design are online schemes. 
The detailed formulation is described as follows.
\begin{itemize}
\item \textbf{Subspace identification}. Estimate $G$ by a closed-form function of the input $u$ and output $y$, i.e., $\hat{G}= f\left(u, y\right)$.
For the sake of statistical analysis and dealing with uncertainties of noise, the derived estimation function \textbf{P1} is supposed to consider the influence of noise,
which is formulated by
\begin{equation}\nonumber
   \textbf{P1:} \ \hat{G}= g\left(u^*, y^*,v,w\right),
\end{equation}
where $y^*$ is the true value of the output and $u^*$ is the equivalent input signal considering the influence of the noise $v,w$ and the input signal $u$ on the system, which is explained in detail in Section \ref{subspace}. 

\item \textbf{Input design}. 
We obtain the optimal input $u$ to minimize the variance of the estimation of $G$ in the subspace method, which is proved (in Section \ref{subspace}) to be equivalent to minimizing the maximum deviation of identification results. 
The problem of the optimal input design is formulated as the following min-max problem.
\begin{equation}\label{min-max0}
\begin{aligned}
\textbf{P2:} \  \min _{u(k_1;k_2)} &\max _{v_i, w_i,v_j,w_j}  \left\|\hat{G}_{i}-\hat{G}_{j}\right\|_{\mathrm{F}} \\
\text { s.t. }  & \hat{G}_{\ell}=g\left(u^*, y^*, v_\ell, w_\ell\right),   \ell = i,j;\\  
& \left\|v\right\|_{\infty} \leqslant v_\mathrm{M}, \left\|w\right\|_{\infty} \leqslant w_\mathrm{M},
\end{aligned}
\end{equation}
where $k_1$, $k_2$ are the start and end time of input design, and $v_\mathrm{M}, w_\mathrm{M}$ are the bound of noise.
We assume that the noise $v$ and $w$ are bounded random variables in this paper.
$\hat{G}_{i}$ or $\hat{G}_{j}$ refers to the identification result under any possible combination $\{v_i, w_i\}$ or $\{v_j, w_j\}$.
\end{itemize}

\begin{remark}
Since the noises $v,w$ are random variables, the identification result $\hat{G}$ is a random variable. 
However, if the bound of the noise $v,w$ is known, the possible value of $\hat{G}$ can be obtained, whose distribution is related to input, output and noise.
The aim of this paper is to minimize the variance of $\hat{G}$ by input design, which provides a more stable identification result when the noise is unknown.
\end{remark}

\subsection{Notations and Assumptions}

In this paper, the lower-case letters $\{x,y,u,v,w\}$ represent vectors and the upper-case letters $\{X,Y,U,V,W\}$ or $\{\textsf{X, Y, U, V, W}\}$ represent the matrices constructed by the respective vectors.
Let $\operatorname{rank}(A)$ be the rank of matrix $A$, $\operatorname{Tr}(A)$ be the trace of matrix $A$,  $A^{\mathrm{L}}$ and $A^{\mathrm{R}}$ be the generalized left and right inverse matrix of matrix $A$ and $A^\dagger$ be the Moore-Penrose inverse matrix of $A$, respectively. 

Define $X(k;h)$ as a transpose vector sequence from $x(k)$ to $x(k+h-1)$, and $\textsf{X}(k;s)$ as the matrix formed by the row arrangement from $x(k)$ to $x(k+s-1)$, i.e., 
\begin{equation}\nonumber
\begin{aligned}
X(k;h) =& [x^{\top}(k),\ x^{\top}(k+1),\  \cdots, \  x^{\top}(k+h-1)]^{\top}, \\
\textsf{X}(k;s) =& [x(k), \ \ \ x(k+1),\ \ \ \cdots,\ \ x(k+s-1)].
\end{aligned}
\end{equation}
A block Hankel matrix formed by vectors from $x(k)$ to $x(k\!+\!h\!+\!s\!-\!2)$ is defined by
\begin{equation}\nonumber
\begin{small}
\mathcal H_{x}(k;h;s)\!=\!\left[\!\begin{array}{cccc}
x(k)\! & x(k\!+\!1) \! \!& \!\cdots\! \! &\! x(k\!+\!s\!-\!1) \!\\
x(k+1)\! & x(k\!+\!2) \! \!& \!\cdots\! \! & \!x(k\!+\!s) \!\\
\vdots\! & \vdots\! \! & \!\ddots\! \! & \!\vdots \!\\
x(k\!+\!h\!-\!1)\! & x(k\!+\!h)\! \! & \!\cdots\! \! &\! x(k\!+\!h\!+\!s\!-\!2)\! 
\end{array}\!\right]\!,
\end{small}
\end{equation}
where $h$ and $s$ determine the dimension of the block Hankel matrix. Then, the matrix composed of the block Hankel matrix $\mathcal H_{y}(k;h;s)$ and $\mathcal H_{u}(k;h+t;s)$ is denoted by
\begin{equation}\nonumber
\mathcal L[y,u] =  \left[\!\begin{array}{l}
\mathcal H_{y}(k;h;s) \\
\mathcal H_{u}(k;h\!+\!t;s)
\end{array}\!\right].
\end{equation}

For an integer $h \geqslant m$, the extended observability matrix $O_c$ and the  extended  controllability matrix $O_b$ are
\begin{equation}\nonumber
\begin{aligned}
&O_{c}(h)=\left[C^{\top}, \   \ (C A)^{\top},  \ \cdots,  \ (C A^{h-2})^{\top},\   (C A^{h-1})^{\top}\right]^{\top}, \\
&O_{b}(h)=\left[A^{h-1} B,\ \ A^{h-2} B,\ \  \cdots,\ \ \ A B,\ \ \ B\right].
\end{aligned}
\end{equation}
A system transformation matrix $T(h)$ is defined as
\begin{equation}\nonumber
\begin{small}
T(h) =\left[
\begin{array}{ccccc}
     0  \\
    CB & 0\\
    CAB &  CB & 0 \\
    \vdots & \vdots & \vdots & \ddots \\
    CA^{h-2}B & CA^{h-3}B & ... & CB & 0
\end{array}\right].
\end{small}
\end{equation}

Finally, the infinite-norm of a matrix $(\cdot)$ is denoted by $\|(\cdot)\|_{\infty} $, and the Frobenius norm is denoted by $\|(\cdot)\|_{\mathrm{F}} $.



\begin{table*}[htbp]
    \centering 
    \caption{A summary of system identification and input design methods}
    \begin{threeparttable}
    \begin{tabular}{ccccccc}
    \toprule
        Paper & Model & Stability & Method & Inputs $u$ & Bias\tnote{[1]} & Variance\tnote{[2]} \\ 
    \midrule
       The Proposed & MIMO & Any & Subspace\tnote{[3]} & Designed & $\mathcal{O}(N^{-\frac{1}{2}})$ & $\mathcal{O}(N^{-1})$\\
       COSMOS \cite{002} & MIMO & Any & Subspace & Gaussian & $\mathcal{O}(N^{-\frac{1}{2}})$ & -  \\
       Zheng \textit{et al}.\cite{022} & MIMO & Any & Subspace & Gaussian & $\mathcal{O}(N^{-\frac{1}{2}})$ & -  \\       
       Oymak \textit{et al}.\cite{8814438} & MIMO & Stable & Subspace & Gaussian & $\mathcal{O}(N^{-\frac{1}{2}})$ & -  \\
       DCP-PEM \cite{001} & Noise-free MIMO & Stable & PEM  & Gaussian & $\mathcal{O}(N^{-\frac{1}{2}})$ & - \\       
       {Manchester} \cite{026} & SISO & Any & Least-squares & Designed & $\mathcal{O}(N^{-\frac{1}{2}})$ & -  \\
       {Casini} \textit{et al}.\cite{027} & FIR SISO & Stable & Least-squares & Designed \tnote{[4]} & $\mathcal{O}(N^{-\frac{1}{2}})$ & $\mathcal{O}(c^{\frac{1}{2^N}})$  \\       
    \bottomrule
    \end{tabular}

    \begin{tablenotes}
        \footnotesize
         \item [1] This column shows the convergence of the error between the identification result and the true value as the amount of data $N$ increases.
        \item [2] This column shows the convergence of the variance of the identification results due to the noise. 
         \item [3] It is an improved method with closed-form estimation of the Markov parameter matrix.  
         \item [4] This work designs the input based on SISO FIR model with binary valued measurements.
      \end{tablenotes}
   \end{threeparttable}
    
   \label{algorithmstable}
    
\end{table*}

The following assumptions are made throughout the paper.\par
\begin{assumption}
System \eqref{sysmodel0} is minimal, and the system order is known.
\end{assumption}
\begin{assumption}
The noise $v,w$ are zero-mean white noises which are independent of the system.
\end{assumption}
\begin{assumption}
The input, output, and noise are bounded. The signal-to-noise ratio of the system is large enough to ignore the influence of the quadratic term of the system noise on the system compared with variables $x, y$ or $u$.
\end{assumption}
Assumptions 1 and 2 are basic guarantees for the feasibility of identification. Since the effect of noise on the system is generally minor, Assumption 3 is reasonable.
These assumptions do not require the stability or a zero initial state, which makes the identification method more general.

\section{Subspace identification and Maximum Identification Deviation}
\label{subspace}
In this section, we first propose a method for estimating the  Markov parameter matrix $G$ to solve \textbf{P1}. 
Then, we analyze the statistic performance of the proposed subspace method. 
We transform the process noise $v$ in system \eqref{sysmodel0} into input noise  for convenience of analysis and obtain the maximum deviation between the identification results.
Finally, we prove that minimizing the maximum identification deviation is equivalent to minimizing the identification variance, which paves the way for the input design to solve \textbf{P2} in next section .

\subsection{Subspace Identification of the Markov Parameter Matrix}


In this subsection, a subspace method is proposed  by  constructing  a  block  Hankel  matrix related to the input and output and eliminating  the system state $x$ without using instrumental variables. 
First, we derive the following lemma.

\begin{lemma}
The block Hankel matrix $\mathcal L[y,u]$, which is composed of $\mathcal H_{y}(k;h;s)$ and $\mathcal H_{u}(k;h+t;s)$, is nonsingular with probability 1.
\end{lemma}
\begin{proof}
Please see Appendix A.
\end{proof}

We start with the identification of a simple condition where the system is noise-free.
We have the following theorem.
\begin{theorem}\label{theMarkov}
Assume that $v = 0, w = 0$, then the Markov parameter matrix $ G $ of system \eqref{sysmodel0} is given by 
\begin{equation}\label{noisefreeidentify}
G(t) =\textsf{Y}(k\!+\!h\!+\!t;s)\mathcal L^{-1}[y,u]\!\left[\!\begin{array}{l}
0 \\
I_{r}
\end{array}\!\right],
\end{equation}
where $s = h\cdot n + (h+t)\cdot p$ and $r = h\cdot n \cdot p$.
\end{theorem}
The constant $k$ determines the start time of the identification.
Constants $h$ and $t$ are arbitrary, and they determine the dimension of system matrices such as $G(t), O_{c}(h)$ and $T(h)$, which means $h$ and $t$ are related to the identification scale and computational complexity.

\begin{proof}
Please see Appendix B.
\end{proof}

Then, we apply Theorem \ref{theMarkov} to the condition that $v, w \neq 0$.
Considering the computation feasibility, we divide the data into $N$ batches for identification. Each batch $i$ is a collection of input and output data at $s$ consecutive time instants with the start time $k = s \cdot i$.
We propose an identification function $f$ by directly applying Theorem  \ref{theMarkov} as follows.
\begin{itemize}
    \item \textbf{ Subspace identification estimation}
\end{itemize}
\begin{equation}\label{subspaceidentify}
\hat{G}(h) = f(u,y)\!=\frac{1}{N}\sum_{i=0}^{N}\! \textsf{Y}(si\!+\!h\!+\!t;s) \mathcal L^{\dagger}[y,u]\!\left[\!\begin{array}{l}
0 \\
I_{r}
\end{array}\!\right],
\end{equation}
where $N$ is the number of batches. 
The proof of consistency, i.e., $\hat G$ converges to the true $G^*$ is given in Theorem \ref{theorem7}.

Estimation \eqref{subspaceidentify} provides a method for identifying matrix $G$ by an expression only related to $y$ and $u$, which simplifies the process of error analysis. Then, the similarity transformation of the system parameter matrices $A,B,C$ can be obtained by the Ho-Kalman Algorithm, where a Hankel matrix $H_G$ is formed based on $G$ and $A,B,C$ are solved via the Singular Value Decomposition (SVD) of $H_G$ \cite{043}.

\subsection{Maximum Identification Deviation}
This subsection analyzes the statistic performance of the proposed subspace identification method.
We first transform the system \eqref{sysmodel0} with process noise into a system with input noise for the convenience of statistic analysis.
Then, the maximum identification deviation is investigated. The optimal solution of the maximum identification deviation is obtained, and it is shown that minimizing the variance is equivalent to minimizing the deviation.

It notes that system \eqref{sysmodel0} has process noise $v$.
The Markov parameters are difficult to identify and thus the statistic performance hard to analyze due to the process noise \cite{022}.
An effective method to tackle this difficulty is to transform the process noise $v$ into input noise which is directly related to $u$.
The transformed system is given by
\begin{equation}\label{sysmodel}
\begin{array}{l}
\left\{\begin{array}{l}
x(k\!+\!1)=A x(k)+B\left(u(k)-e(k)\right) \\
y(k)=C x(k)+w(k),
\end{array}\right. 
\end{array}
\end{equation}
where $e(k) \in   \mathbb{R}^p,w(k) \in   \mathbb{R}^m$ represent the input and output noise, respectively.  
Then, we prove that the transformation is reasonable by the following theorem.
\begin{theorem}\label{theTransformation}
Suppose Assumptions 1-3 hold, then there exists a bounded zero-mean variable $e$, so that the influence of $v$ in \eqref{sysmodel0} is equivalent to the influence of $e$ in system \eqref{sysmodel}.
\end{theorem}
\begin{proof}
Please see Appendix C.
\end{proof}
Theorem \ref{theTransformation} implies that the process noise $v$ can be transformed into the input noise $e$ in the equivalent system \eqref{sysmodel}. 

Thus, substituting $v$ by $e$, we derive the solution to \textbf{P1}.
\begin{itemize}
    \item \textbf{Solution to P1}
\end{itemize}
\begin{equation}\nonumber
     \hat{G}(h) = g\left(u^*, y^*, e, w \right)  = f(u^*+e,y^*+w),
\end{equation}
where $f$ is given in \eqref{subspaceidentify}.

Then, we analyze the statistic performance of our method. We focus on the maximum identification deviation $J(u)$, which is defined by
\begin{itemize}
    \item \textbf{Maximum identification deviation}
\end{itemize}
\begin{equation}\nonumber
    J(u)=\max _{e_i, w_i,e_j,w_j}\left\|\hat{G}_{i}-\hat{G}_{j}\right\|_{\mathrm{F}}.
\end{equation}

First, we make the following notations in this subsection.
We define $y^{*}=y-w$ and $u^{*}=u-e$, which means $y^{*}$ and $u^{*}$ are input and output signals that are not affected by uncertainty of noise.
Let $d = k\!+\!h\!+\!t$, $\textsf{W}_\ell(d;s)$ be the matrix formed by the vector from $w_\ell(d)$ to $w_\ell(d+s-1)$. Let $\mathcal L^{-1}_{\beta_\ell} = \mathcal L^{-1}[y^*\!+\!w_\ell,u^*\!+\!e_\ell]$.
Denote the elements of the $i$-th row and $j$-th column of the square matrices $\mathcal L^{-1}[y^*,u^*]$ as $\alpha(i,j)$,  $  \mathcal  L^{-1}_{\beta_\ell}$  as $\beta_\ell(i,j)$ and  $\mathcal L[w,e]$ as $p(i,j)$, respectively. 
Denote $\bm{\beta}$ as $(\beta_1 - \beta_2)$, $\bm{w}$ as $(w_1 - w_2)$ and $\bm{e}$ as $(e_1 - e_2)$. 
Define $\delta = \max\{e_M, w_M\}$.
It follows that  $\| \bm{e} \|_{\infty}\leqslant 2\delta$, $\| \bm{w}\|_{\infty} \leqslant 2\delta$. 
Define $\bm{w}^+ = \bm{w}(d+r;s-r)$ and  $\bm{w}^- = \bm{w}(1;d+r)$.
We have $\bm{w} = [\bm{w}^-;\bm{w}^+]$.

Next, we use the above notations to expand the expression of $J(u)$ and split it. We have the following theorem.
\begin{theorem} \label{theorem3}
Solving $J(u)$ is equivalent to solving the two sub-problems $J_1$ and $J_2$, where
\begin{equation}\label{J1}
J_{1} =\max _{\bm{w}^+}\sum_{j=r+1}^{s}\left(\sum_{i=r+1}^{s} \bm{w}(d+i-1) \alpha(i, j)\right)^{2},
\end{equation}
\begin{equation}\label{J2}
J_{2}=\max _{\bm{e},\bm{w}^-}\sum_{j=r+1}^{s}\left(\sum_{i=r+1}^{s} \bm{\beta}(i, j) y^*(d+i-1)\right)^{2}.
\end{equation}
\end{theorem}
\begin{proof}
By \eqref{noisefreeidentify}, we have
\begin{equation}\label{rg}
\begin{small}
\begin{aligned}
 G_{\ell}  \! =\!\left[\textsf{Y}^*(d;s)\!-\!\textsf{W}_\ell(d;s)\right]  \mathcal L^{\!-\!1}[y^*\!+\!w_\ell,u^*\!+\!e_\ell]\!\left[\!\begin{array}{l}
0 \\
I_{r}
\end{array}\!\right],  \ \ell=i,j.
\end{aligned}
\end{small}
\end{equation}
Since
\begin{equation}
\frac{\partial \mathcal L^{-1}_{\beta_\ell}}{\partial p(l, t)}=  \mathcal L^{-1}_{\beta_\ell} \frac{\partial \mathcal L_{\beta_\ell}}{\partial p(l, t)} \mathcal  L^{-1}_{\beta_\ell},
\end{equation}
we have 
\begin{equation}\label{paB}
\frac{\partial \beta_\ell(i, j)}{\partial p(l, t)}= (\alpha(t, j) - p(t,j))( \alpha(i, l) - p(i,l)).
\end{equation}
Note that in the Taylor expansion of $\beta_\ell$ with respect to $p$, the $p(t,j)$ and $p(i,l)$ in the partial derivative \eqref{paB} corresponds to a quadratic term of $p$, which can be ignored by Assumption 3. Then, $\beta_\ell$ is a proportional function of $p$. Denote $\beta_\ell$ as $\epsilon p$. It follows that $w\beta_\ell$ can be ignored when $\epsilon$ is relatively small.
Hence, the expansion of $J$ is reduced to
\begin{equation}\label{J}
\begin{small}
\begin{aligned}
J^2(u)& \!=\!\max _{e_1,w_1,e_2,w_2}\! 
\sum_{j=r+1}^{s}\!\left(\! \sum_{i=r+1}^{s}\left(\beta_{1}(i, j)\!-\!\beta_{2}(i, j)\right) y^*(d+\!i\!-\!1)  \right.\\  &\left.+ \sum_{i=r+1}^{s}\left(w_{1}(d+\!i\!-\!1)-w_{2}(d+\!i\!-\!1)\right) \alpha(i, j) \right)^{2}.
\end{aligned}\end{small}
\end{equation}

By Assumption 2 and Theorem \ref{theTransformation}, $\bm{e},\bm{w}$ are zero-mean white noise independent of the system. Hence, $J_1$ and $J_2$ are two independent problems. Therefore, expanding the expression of $J$ in \eqref{J}, neglecting the quadratic term by Assumption 3, it holds that solving $J$ is equivalent to solving $J_1$ and $J_2$.
\end{proof}

From \eqref{J1}, it follows that
\begin{equation}\label{J12}
J_{1} = \max _{W} W^{\top}(d;s) H W(d;s),
\end{equation}
where $H(i, j)=\sum_{k=r+1}^{s} \alpha(k, i) \alpha(k, j).$

The maximization problem \eqref{J12} is a quadratic programming problem with a positive semi-definite Hessian matrix $H$, which makes it an NP-Hard, non-convex problem\cite{038}.

Considering that $J_{1}$ is a quadratic function with a positive first coefficient for each variable $\bm{w}$, the optimal $J_{1}$ is obtained only when all $\bm{w}$ have reached the bound $\delta$ or $-\delta$. When $s$ is relatively small, \eqref{J12} can be solved by enumeration. However, for the general case, it is necessary to perform a relaxation.

Denote $Q_W$ as $W(d;s) W^{\top}(d;s)$ and $\mathcal S$ as the set of positive semi-definite matrices with the same dimension as $Q_W$, respectively. Then, \eqref{J1} is equivalent  to
\begin{equation}\label{J1p}
J_{1}=\max _{Q_W \in \mathcal C} \operatorname{Tr}(Q_W H),
\end{equation}
where
$\mathcal C:\left\{Q: Q \leqslant 4 \delta^{2} I, Q \in \mathcal S, \operatorname{rank}(Q)=1\right\}.$
Define a relaxed convex feasible set $\mathcal{C}_{\text{relax}} : \left\{Q: Q \leqslant 4 \delta^{2}, Q \in \mathcal S \right\}$
and the relaxed problem
\begin{equation}\label{J1r}
Q^* =\arg \max _{Q \in \mathcal C_{\text{relax}}} \operatorname{Tr}(Q H).
\end{equation}
Then, the relaxed problem \eqref{J1r} can be efficiently solved by semidefinite programming \cite{026}.
\begin{theorem}\label{theorem4}
Denote $q$ as the eigenvector corresponding to the largest eigenvalue of the solution $Q^*$ of \eqref{J1r}. Then, $q$ is also the optimal solution of \eqref{J1}, i.e., $q = w^{+*}$. 
\end{theorem}

The proof of Theorem \ref{theorem4} can be directly obtained by referring to Section 4 in \cite{026}, since the transformed sub-problem \eqref{J1r} is a low dimensional case in \cite{026}. Hence, the solution of \eqref{J1} is obtained from Theorem  \ref{theorem4}.

For the solution of $J_2$, we denote the objective function of $J_2$ by $f_{j_2}$,  then,
\begin{equation}\label{J2p}
\frac{\partial f_{j_2}}{\partial p(l, t)}=2 \sum_{j=r+1}^{s} \sum_{i=r+1}^{s} \bm{\beta}(i, j) y^*(d+i-1) \alpha(t, j) \alpha(i, l).
\end{equation} 
Hence, one infers that the second-order partial derivative of $ f_{j_2}$ can be regarded as a constant, i.e., the Hessian matrix $ H_2 $ of the multivariate function $ f_{j_2}$ with respect to $p$ is a constant matrix. According to \eqref{paB} and \eqref{J2p}, the Taylor expansion of $ f_{j_2}$ related to $p$ does not have a first-order term. Then, $J_2$ can be transformed into the following quadratic problem
\begin{equation}\label{J2pp}
J_{2} = \max _{P} P^{\top} H_2 P,
\end{equation}
where $P$ denotes the sequence composed of $p(i,j)$ and $ H_2 $ is an Hessian matrix which can be solved from \eqref{J2p}. 
Then, the solution of $J_2$ is provided by Theorem \ref{theorem4}. By solving \eqref{J1p} and \eqref{J2pp}, we obtain the optimal solution of \eqref{J}, which is the maximum identification deviation.

To illustrate that  minimizing the identification deviation is equivalent  to  minimizing the variance, we first define the variance of $L$ times of identification as
\begin{equation}\label{var}
    \mu = \sum_{\ell=1}^{L}\|G_\ell - \bar{G}_\ell \|_{\mathrm{F}}^2,
\end{equation}
where $ \bar{G}_\ell = \frac{1}{L} \sum_{\ell=1}^{L}G_\ell$.

Similar to \eqref{J}, $\mu$ is also a quadratic form of the noise with coefficients $y$ and $u$.
Considering that the noise is independent, and $\bm{\beta}$ is a linear function of the noise from \eqref{paB}, then, the change of $y$ and $u$ does not affect the distribution of identification results. 
Hence, minimizing the variance $\mu$ is equivalent to minimizing the maximum deviation $J$.

\section{Input Design for Minimizing Identification Variance}\label{inputdesignsection} 

This section describes the method of input design, which is shown in Algorithm 1, and provides the analysis of the performance of the input design algorithm.
We take the subspace method proposed in this paper to identify the system, then design $u$ to minimize the maximum deviation.
The convergence of the identification error and the identification deviation obtained by the proposed method are proved. 

\subsection{Input Design Algorithm}

This subsection describes the input design algorithm as shown in Algorithm 1. The algorithm is divided into several parts, namely, the initialization, the identification of system matrices, the computation of feasible set of input and the design of input for solving \textbf{P2}.

\subsubsection{Initialization}
In initialization, time is sampled discretely into multiple intervals with $s$ time points in each interval. We set the start time $k$ and the number of groups of data $N$. 
The maximum of the quadratic term of the noise which can be ignored is set as $\epsilon$. 
For security, the bound of the infinite-norm of $y$ and $u$ are determined as $y_\mathrm{M}$ and $u_\mathrm{M}$.
We denote the maximum noise as $\delta$ by Assumption 2. 
At last, we initialize a sequence $Y(0;s\!+\!1)$ and $U(0;s\!+\!1)$ without special design to preliminarily identify the system.

\subsubsection{Identification of System Matrices}
At each $i$-th iteration, we use the subspace method to identify $\hat{A}_{i}$, $\hat{B}_{i}$ and $\hat{C}_{i}$, which are required in \eqref{J} for obtaining the maximum identification deviation. We estimate $\hat{G}(t)$ from \eqref{subspaceidentify} by the latest data $Y(i;s)$ and $U(i;s)$, then compute $\hat{A}_{i}$, $\hat{B}_{i}$, $\hat{C}_{i}$ by the Ho-Kalman Algorithm. 

\subsubsection{Feasible Set of Input}
At $i$-th iteration, the input needs to be designed is $u(s\!+\!i)$.
The input is supposed to ensure the system output $y$ is within the safety range.
Since we do not know the true value of $y$ of next time period, we use the prediction of $y$ instead, which is defined as follows.
\begin{equation}\label{esY}
\begin{aligned}
\hat{Y}(i+1;s) = & \ \hat{O}_c \hat{A}^t O_c^{\mathrm{L}} Y(i;s) + \hat{O}_c \hat{A}^t O_c^{\mathrm{L}} \hat{G}_{i}(t) U(i\!+\!1;s) \\&+ \hat{O}_c \hat{O}_b U(i;s) + \hat{G}_{i}(t) U(i\!+\!1;s).
\end{aligned}
\end{equation}
We keep the real $y$ safe by leave a certain margin of $y_\mathrm{M}$. 

Additionally, $u$ is designed so that $\alpha$ and $\beta$ are bounded within $\alpha_\mathrm{M}$ to prevent excessive errors.
Denote the possible value of $\alpha$ and $\beta$ with the largest norm that can be obtained under bounded noise as $\tilde{\alpha}$.

Therefore, the feasible set $\mathcal{U}$ is given by
\begin{equation}\label{feset}
\mathcal{U}\!:\! \left\{u \!:\! \| \hat{y} \|_{\infty}\!\leqslant \!y_{\mathrm{M}},  \| u \|_{\infty} \!\leqslant\! u_{\mathrm{M}}, \delta \| \tilde{\alpha}  \|_{\infty}^2  \!\leqslant\! \epsilon,  \| \tilde{\alpha} \|_{\infty} \!\leqslant\! \alpha_{\mathrm{M}} \right\}.
\end{equation}

Next, we prove the existence of the feasible set $\mathcal{U}$ in our algorithm.
Considering that the influence of $u$ on an LTI system may have time delay,  Algorithm 1 applies the idea of predictive control to ensure safety when solving the feasible set and designing $u$. 
We find the feasible set of $U(s+i ;m)$ instead of $U(s+i ; 1)$ (i.e., $u(s+i)$) to ensure that the output value of the system $y$ does not exceed the bound.
Only the first $u(s+i)$ is designed and used in the real system. 
The next theorem shows that using the idea of predictive control to design the input signal ensures the system within the safety constraints, i.e., the feasible $u(s+i)$ always exists.

Denote $\mathcal{U}_i$ as the set of sequence of $U(k; h+i)$ and $U(k; h+M)$ such that the system is within the safety constraints from the time $k$ to $k\!+\!h\!+\!i\!-\!1$. Obviously, it follows that $\mathcal{U}_{i+1} \subseteq \mathcal{U}_{i} \subseteq \mathcal{U}_{i-1} \subseteq \cdots \subseteq \mathcal{U}_{0}$. Then, we have the following theorem.
\begin{theorem} \label{theorem5}
$\mathcal{U}_{i+1} \neq \varnothing$ iff $\mathcal{U}_i \neq \varnothing$.
\end{theorem}
\begin{proof}
Please see Appendix D.
\end{proof}
Theorem \ref{theorem5} shows that if there exists an initial input signal sequence in  $\mathcal{U}_{0}$, the feasible input can always be designed in $\mathcal{U}_{i}$. Hence, during each iteration in the input design algorithm, we design the feasible first-step $U(i  s; 1)$ and thus, by Theorem \ref{theorem5}, there exists feasible $u$ in next iteration. Therefore, we ensure the safety constraints through the method of predictive control.

\subsubsection{Design of Input}
Finally, we design the input signal in the feasible set.
By the transformation in Theorem \ref{theTransformation}, the input design problem \textbf{P2} is equivalent to the following one,
\begin{equation}\label{min-max}
\begin{aligned}
 &\textbf{P3.} \
\min _{u}  J(u)  \\
\text { s.t. }  &  J(u)=\max _{e_i, w_i,e_j,w_j}\left\|\hat{G}_{i}-\hat{G}_{j}\right\|_{\mathrm{F}} \\
& \hat{G}_{\ell}=g\left(u^*, y^*, e_\ell, w_\ell\right), \ell = i,j;\\
& \left\|e\right\|_{\infty} \leqslant e_M, \left\|w\right\|_{\infty} \leqslant w_M,
\end{aligned}
\end{equation}
where $y^{*}=y-w$, $u^{*}=u-e$, and 
\begin{equation}
    g\left(u^*, y^*, e_\ell, w_\ell\right)  = f(u^*+e_\ell,y^*+w_\ell).
\end{equation}


By \eqref{J1}, \eqref{J2} and \eqref{J}, \textbf{P3} is equivalent to minimize the squared maximum deviation of estimation of $G$ at $N$ iteration, which is defined as
\begin{equation}\label{costab}
\begin{aligned}
    \min _{u} \max _{e,w} \sum_{j=r+1}^{s}&\left(\sum_{i=r+1}^{s}\beta(i,j)y(d+i-1) \right. 
    \\ &\left.+\sum_{i=r+1}^{s}\alpha(i,j)w(d+i-1) \right)^2,
\end{aligned}
\end{equation}
where $\alpha(i,j)$ and $\beta(i,j)$ are elements in the inverse matrix of the matrix related to $u$, see the previous section for details.

Next, we give the method to obtain the analytical solution of $u$ in \textbf{P3}.
Define the matrix formed by $\alpha(i,j)$ as $ \left[\begin{array}{cc}
    Y & y\\
    u_0^\mathrm{T} &  u
\end{array}\right]^{-1}$.
Denote $(u - u_0^\mathrm{T}Y^{-1}y)^{-1}$ by $u_2$, $Y^{-1}y$ by $y_1$. We have
\begin{equation}\label{YU-1}
\left[
\begin{array}{cc}
    Y & y\\
    u_0^\mathrm{T} &  u
\end{array}\right]^{-1}
=\left[
\begin{array}{cc}
    Y^{-1} +y_1 u_2 u_0^\mathrm{T}Y^{-1} & -y_1 u_2\\
    -\bm{u}^\mathrm{T}Y^{-1}u_2 &  u_2
\end{array}\right].
\end{equation}
By \eqref{YU-1}, we can obtain the relationship between $\alpha(i,j)$ and $u$.
Since $u$ is the last item of the matrix formed by $\alpha(i,j)$, $Y, y, u_0^\mathrm{T}$ and $u$ are independent. Hence, $Y, y, y_1, u_0^\mathrm{T}$ and $u_2$ are independent.
Similarly, when we fixed the noise $e_i, w_i, e_j, w_j$, the relationship between $\beta(i,j)$ matrix and $u$ can be formulated as
\begin{equation}\label{YU-2}
\left[
\begin{array}{cc}
    Y_i & y_i\\
    \bm{u}_i^\mathrm{T} &  u_i
\end{array}\right]^{-1}- \left[
\begin{array}{cc}
    Y_j & y_j\\
    \bm{u}_j^\mathrm{T} &  u_j
\end{array}\right]^{-1}
\end{equation}
Similar to \eqref{YU-1},   $Y_\ell, y_\ell, y_\ell, \bm{u}_\ell^\mathrm{T}$ and $u_\ell$ are independent, $\ell = i,j$.
Combining \eqref{costab}, \eqref{YU-1} and \eqref{YU-2}, substituting the optimization variable of the input design problem into $u_2$, we have the equivalent input design problem as follows,
\begin{equation}\label{costu2}
\begin{aligned}
    \min _{u_2} J_0(u_2) = \min _{u_2} \max _{e,w} \sum_{j=r+1}^{s}&\left( F(y\!,\!u\!,\!e\!,\!w)u_2 \!+\! c(y\!,\!u\!,\!e\!,\!w) \right)^2,
\end{aligned}
\end{equation}
where $F(y,u,e,w), c(y,u,e,w)$ are functions only related to $y, u, e, w$ and not related to $u_2$.

For the solution of \eqref{costu2}, we have the following lemma.
\begin{lemma}
$J_0(u_2)$ is a convex function related to $u_2$.
\end{lemma}
\begin{proof}
Please see Appendix E.
\end{proof}
By Lemma 2, we can obtain the optimal solution of $u_2$ through the gradient descent method, and the global optimality can be guaranteed.
To speed up the convergence rate, considering that $J_0(u_2)$ is a quadratic function form of $u_2$ when $y,u,e,w$ are fixed, we can quickly initialize the gradient descent method.
The algorithm for solving the optimal solution of $u_2$ is summarized as follows.
\begin{itemize}
    \item \textbf{Gradient Descent for Input Design}
\end{itemize}
\begin{enumerate}[i)]
\item Initialize: set $e_i,w_i,e_j,w_j$ to zero in $J_0(u_2)$ and solve the optimal solution $u^*$ of the quadratic function $J_0(u_2)$.
\item Gradient calculation: consider the change of the function value of $J_0(u_2)$ when $u_2$ changes a small range $\delta_u$ as the gradient of $J_0(u_2)$. 
The value of $J_0(u_2)$ at each point $u(i)$ is given by Theorem \ref{theorem3} and Theorem \ref{theorem4}.
Denote the gradient at $u(i)$ as $\nabla_i$.
\item Gradient descent: Perform gradient descent until convergence. At the $i$-th iteration, $u_2(i+1) = u_2(i)-\lambda_i \nabla_i$, where $\lambda_i$ is the learning rate satisfying $\sum_{i=1}^{\infty}\lambda_i = \infty, \ \sum_{i=1}^{\infty}\lambda_i^2 < \infty$.
\end{enumerate}

After the gradient descent method, we obtain $u_d$ through $u_2$ and derive the matrix $G$ by the subspace identification function. 
Repeating the process of system identification and input design, the identification results are continuously updated while designing input to minimize the variance.

\subsection{Performance Analysis}
In this subsection , we analyze the optimality of the design of input, and the convergence and stability of the proposed identification method.

First, we analyze the relationship between the designed input signal and the optimal input for minimizing the maximum deviation of identification results in \eqref{costab}.
Denoting the designed input in Algorithm 1 by $u_{\text{Alg}}$, then we have the following theorem.
\begin{theorem}\label{theorem6}
$u_{\text{Alg}}$ is the optimal solution to problem \eqref{costu2} when $u_{\text{Alg}} \in \mathcal{U}$.
\end{theorem}
\begin{proof}
By Lemma 2, problem \eqref{costu2} is a convex optimization problem when we do not consider the constraints. Hence, when $u_{\text{Alg}} \in \mathcal{U}$, the input signal derived by gradient descent, $u_{\text{Alg}}$, is the optimal solution.
\end{proof}

\begin{algorithm}[t]
		\small
		\LinesNumbered
		\caption{Input Design Algorithm}\label{Inputdesignal}
		{\bf Initialize:}
        system model \eqref{sysmodel},  initial sequence of $y$ and $u$, $\epsilon$, $\delta$, $h$, $t$, $N$, $y_{\mathrm{M}}$, $u_{\mathrm{M}}$, $\alpha_\mathrm{M}$

		\While{$i < N$}
		{
		\KwIn{latest output data $y_i$.}
	
		compute $\hat{G}(t)$ by \eqref{subspaceidentify}\\
		obtain $\hat{A}_{i},\hat{B}_{i},\hat{C}_{i}$ by the Ho-Kalman Algorithm \cite{043}\\
		find the feasible set $\mathcal{U}$ of $U(s+i; m)$ by \eqref{feset}\\
		use gradient descent to solve \textbf{P3} and get $u(s+i)$ in $\mathcal{U}$.\\ 
		\KwOut{$\hat{A}_{i},\hat{B}_{i},\hat{C}_{i}$, designed input $u(s+i)$.
		}
		}
\end{algorithm}

Next, we analyze the convergence of the identification error of Algorithm 1, which implies the accuracy of the proposed identification method.
Let $\Delta G$ be the square of the identification error of the Markov parameter matrix $G$, i.e.,
\begin{equation}\nonumber
    \Delta G = \left\|\hat{G}(t)-G^{*}(t)\right\|_{\mathrm{F}}^{2}, 
\end{equation}
where  $G^{*}$ is the true value of the Markov parameter matrix and $\hat{G}$ is the estimation obtained by Algorithm 1.
Then, we have the following theorem.
\begin{theorem}\label{theorem7}
The identification error satisfies 
\begin{equation}\nonumber
\lim\limits_{N\to\infty} \Delta G= 0;
\end{equation}
Furthermore, if the noise obeys Gaussian distribution, then 
\begin{equation}\nonumber
    \Pr\left\{\Delta G \leqslant \frac{c_1}{N}(c_2 + \tau)^2\right\}>1 - 2 e^{-\tau^2/2},
\end{equation}
holds for $\forall \tau \geqslant 0$, where $c_1, c_2$ are constants determined by the system model.
\end{theorem}
\begin{proof}
Please see Appendix F.
\end{proof}

Theorem \ref{theorem7} shows that the estimation error of our algorithm converges at a speed of $1/\sqrt{N}$ and converges to 0 in the infinite time domain.
Actually, the designed input obtains a better result with a lower amount of data compared with the white noise input as shown in simulations.

Finally, we investigate the stability of the identification results derived by Algorithm 1, which is the goal of input design. 
\begin{theorem}\label{theorem8}
The maximum identification deviation of the proposed method with input design converges to 0 at a speed of $\mathcal{O}(1/N)$.
\end{theorem}

\begin{remark}
Theorem \ref{theorem8} is only applicable to the condition that the input can be designed, where we can ensure the maximum value of $\alpha$, i.e., $\alpha_{\mathrm{M}}$ is a constant.
If the input is not designed, the speed of convergence of maximum identification deviation is slower, e.g., $\mathcal{O}(1/\sqrt{N})$.
\end{remark}

\begin{proof}
Please see Appendix G.
\end{proof}

Theorem \ref{theorem8} shows that the input design process in our proposed method have smaller identification deviation and identification variance.
It shows that input design contributes to a more stable and more efficient identification result, and that the use of noise input slows the convergence speed of the maximum identification deviation.
This is also reflected in the simulation in the next section.
Overall, by Theorem 6-8, it implies that Algorithm 1 designs the optimal input signal to achieve identification with the error converges to zero and the minimum variance.


\section{Numerical simulation and Experiments}\label{Sim}

This section uses a numerical simulation to compare the performance of our input design method with a white noise input method and an PEM-based method. 
The simulation result verifies the effectiveness of Algorithm 1 in this paper.

We randomly generate an SISO system model of order 4 and convert it into a controllable canonical form. 
The conversion is reasonable because the subspace method only identifies a similar transformation of the system, and the conversion does not change the matrix $G$. 

The controllable canonical form of the model is
\begin{equation}\nonumber
\begin{small}
\begin{aligned}
&A = \left[\begin{array}{cccc}
0 & 1 & 0 & 0 \\
0 & 0 & 1 & 0 \\
0 & 0 & 0 & 1 \\
-1.23 & -2.17 & -1.42 & -1.21
\end{array}\right],\\
&B =\left[0\ 0\ 0\ 1\right]^{\top}, \quad C = [0.82\ 0.17\ -0.28 \ 0.27].
\end{aligned}
\end{small}
\end{equation}
The constraints of the SISO system are $\delta = 0.05$, $y_{\mathrm{M}}  = 100$, $u_{\mathrm{M}} = 10$.
The noise is generated as a white noise sequence obeying uniform distribution.
We take the Frobenius norm of the Markov parameter matrix $G$ as the index to evaluate the system identification result, which is consistent with the optimization goal of \eqref{min-max}.

In Figure 3, we consider the scenario of identifying a running system.
We use $x(0) = [0;0.5;0.3;1]$ to generate an initial sequence of length 44 and apply the input design method in this paper for system identification, compared with the white noise input and the input which maximizes the Fisher information matrix based on the PEM\cite{026}. 
For generality, we conduct 100 Monte Carlo trials under random noise generation and record the average error when the number of batches of data increases.
It is observed that as the number of data increases, the identification errors of the three identification methods decrease.
The error of the proposed method and the PEM-based method are always within the upper bound defined by Theorem \ref{theorem5} ($t = 5$), and much smaller than white noise input. 
In 80 batches of data, the error of the proposed method and the PEM-based method are 28.3\% and 40.3\% of the white noise input, respectively.
Note that the performance of the PEM-based method is similar to the proposed method in large time domain.
However, the error of the PEM-based method tends to be large when the amount of data is relatively small.
It implies that the proposed method has advantages in fast identification in a short time. 
This is because that the proposed  method makes full use of the subspace expression in Section \ref{subspace} and the safety range of $y$, as shown in Figure 3. 
In addition, the identification error of the PEM-based method fluctuates when the data set is not large, and the reliability of the PEM-based result is not as high as the proposed method.

Figure 4 and 5 provides intuitive evidence that our method reduces the maximum deviation. In Figure 4, we consider the distribution of the results of multiple identifications of systems with the same parameter matrix and different noises.
We made simulations with the three types of input in the case of 70 and 250 batches of data. 
The identification results of 70 batches taking random input or the PEM based input have large variance and many outliers.
The proposed method needs only 70 batches to achieve better identification result than the 250 batches of data that traditionally uses random input.
Figure 5 shows the maximum deviation of identification results of the proposed method with designed input signal and the method which uses the white-noise input. 
It shows that when white noise is used as the input, the maximum deviation of the identification result is large, and the convergence speed is slower over time. 
In 100 Monte Carlo experiments, the convergence of maximum deviation of the white noise input is not stable.
Moreover, the maximum deviation is no tendency to converge to zero.
In contrast, the proposed input design method makes the maximum deviation quickly converge, and always maintains a faster convergence trend to zero.

To demonstrate that our algorithm is effective for general linear time-invariant systems, in Figure 6, we randomly generate  three linear systems and apply our algorithm. The model of the first system is
\begin{equation}\nonumber
\begin{small}
\begin{aligned}
&A = \left[\begin{array}{cccc}
-23.00	& -13.25&	-20.20 &	-14.63\\
14.26 &	8.13 &	13.46 &	8.74 \\
8.12 &	4.31  &		5.36  &		6.37 \\
13.85 &	8.51  &		11.28 &		8.69
\end{array}\right],\\
&B =\left[12.72\
-8.14 \
-2.38\
-7.43 \right]^{\top}, \\
&C = [-0.58\	-0.99\	-0.10\ 0.06 ].
\end{aligned}
\end{small}
\end{equation}
The other systems are also randomly generated.
It can be seen from Figure 6 that our algorithm can converge and obtain stable identification results under the random linear systems.

The output signal, the prediction of output and the error of prediction during the identification process are given by Figure 7. 
It shows that the model in this paper accurately predicts the value of $y$, and that the output signal fluctuates within a safe range. 
Since the safety constraint is soft, the output signal violates the constraint in a short time. 
Strict restrictions can be achieved by setting a smaller constraint tolerance.

These simulation results demonstrate the effectiveness of the proposed input design algorithm.

\begin{figure}[t]
  \centering 
  \setlength{\abovecaptionskip}{0.1cm}
    \label{sae1} 
    \includegraphics[width=0.45\textwidth]{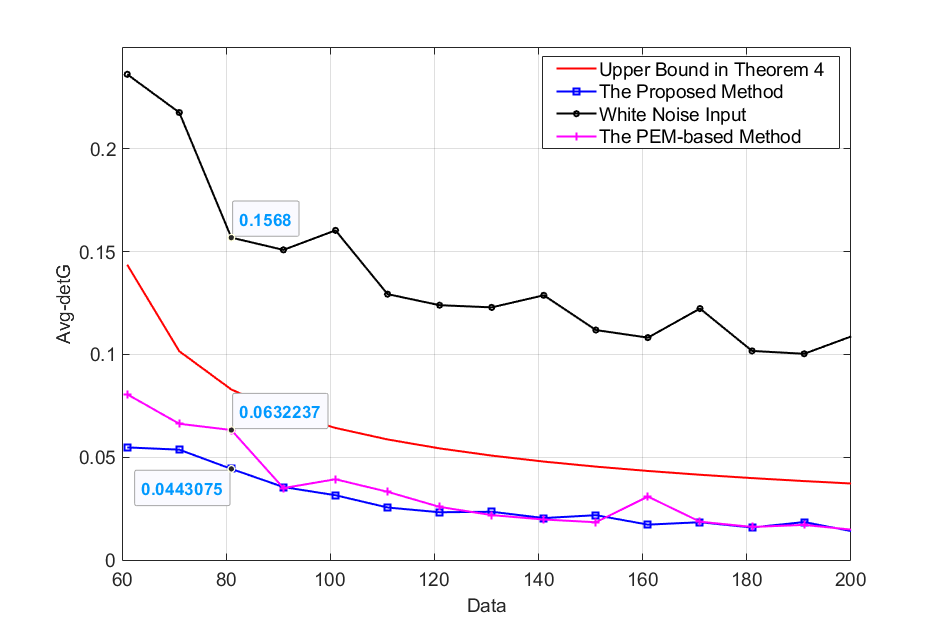} 
  \caption{The average of the error of the identification result ($\|\hat{G}- G^*\|_{\mathrm{F}}$), obtained from 100 Monte Carlo runs with random noise.} 
  \vspace*{-10pt}
\end{figure}

\begin{figure}[t]
  \centering 
  \setlength{\abovecaptionskip}{0.1cm}
    \label{sae2}
    \includegraphics[width=0.45\textwidth]{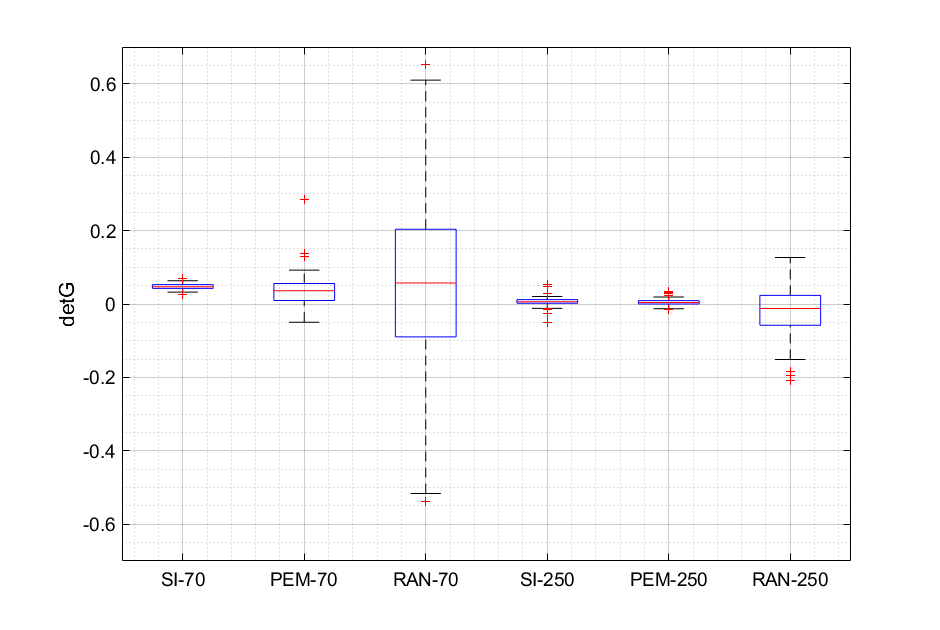} 
  \caption{Box-plot of the distribution of identification error, which compares the distribution of the subspace input design algorithm proposed in this paper (SI), the input design method based on the PEM (PEM) and the random white noise input (RAN) under 70 and 250 sets of data.} 
  \vspace*{-10pt}
\end{figure}

\begin{figure}[t]
  \centering 
  \setlength{\abovecaptionskip}{0.1cm}
    \label{sae3} 
    \includegraphics[width=0.45\textwidth]{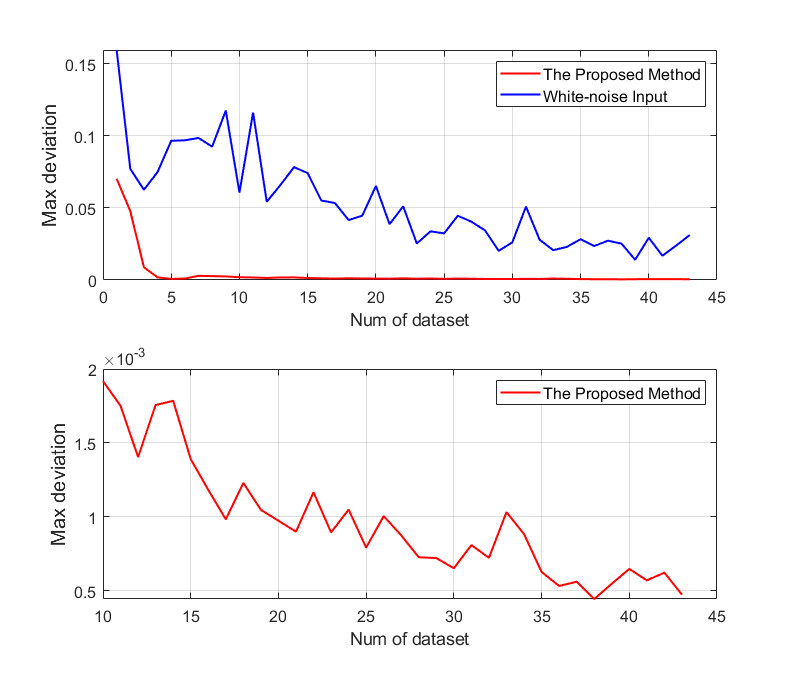} 
  \caption{The maximum deviation of identification results of the proposed method with designed input signal (the red line) and the method which uses the white-noise input (the blue line), obtained from 100 Monte Carlo runs with random noise.
  } 
  \vspace*{-10pt}
\end{figure}

\begin{figure}[t]
  \centering 
  \setlength{\abovecaptionskip}{0.1cm}
    \label{sae5} 
    \includegraphics[width=0.45\textwidth]{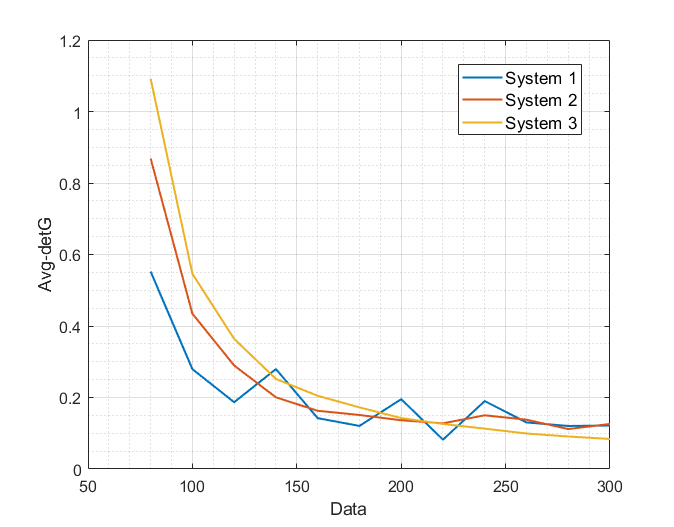} 
  \caption{The average of the identification error of several random systems, obtained from 100 Monte Carlo runs with random noise.} 
  \vspace*{-10pt}
\end{figure}

\begin{figure}[t]
  \centering 
  \setlength{\abovecaptionskip}{0.1cm}
    \label{sae5} 
    \includegraphics[width=0.45\textwidth]{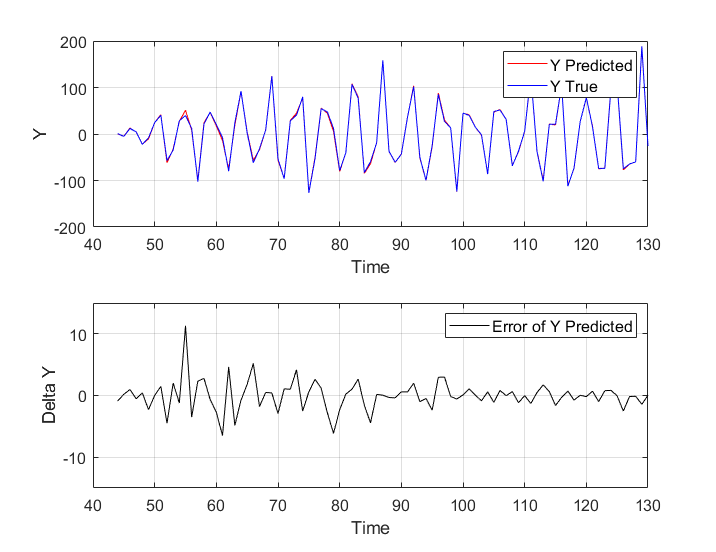} 
  \caption{The output (the blue line) and its prediction (the red line) and the error of the prediction (the black line) during the identification process.} 
  \vspace*{-10pt}
\end{figure}

\section{Conclusion}\label{conclusionfi}
In this paper, we propose an improved subspace method with a closed-form and consistent estimation of the system matrix.
Then, we derive an input design algorithm to deal with the uncertainty of noise in system identification while ensuring safety.
Our research provide a feasible way for the observer to tackle the difficulty of statistic analysis in subspace identification, and to achieve more accurate and more stable identification results.
We investigate the state-space model of the LTI system and identify the Markov parameter matrix via expressing it as an explicit function of input and output. 
We investigate the maximum identification deviation of the identification result.
Finally, an input design algorithm is presented to achieve more accurate and stable identification results via minimizing the identification variance.
Simulation results are provided to illustrate the effectiveness of the proposed method.

Future directions include
i) considering the case where the input signal can only be partly designed;
ii) investigating general models such as nonlinear systems;
iii) extending the application of input design to a locally observable network.

\section*{Appendix}
\subsection{Proof of Lemma 1}

\begin{proof}
First, we prove that the Hankel matrices $\mathcal H_{y}(k;h;s)$ and $\mathcal H_{u}(k;h\!+\!t;s)$ have full row rank with probability 1.

Consider the Hankel matrix $\mathcal H_{y}(k;h;s)$. This matrix has $h$ rows and $s$ columns.
Since $\mathcal H_{y}(k;h;s)$ is formed by vectors from $y(k)$ to $y(k\!+\!h\!+\!s\!-\!2)$ and each $y$ is the output signal with independent random output noise, the matrix $\mathcal H_{y}(k;h;s)$ has $(h+s-1)$ random variables.

Suppose that $\mathcal H_{y}(k;h;s)$ is a singular matrix, then the vectors formed by each row of  $\mathcal H_{y}(k;h;s)$ are linearly related, i.e., there exist $k_1,k_2,\cdots,k_{h-1}$, such that
\begin{equation}\label{Hsigular}
\begin{small}
\begin{aligned}
     &y(k) = k_1 y(k\!+\!1) \!+\! k_2 y(k\!+\!2) \!+\! \cdots \!+\!  k_{h-1} y(k\!+\!h\!-\!1)\\
     &y(k\!+\!1) = k_1 y(k\!+\!2) \!+\! k_2 y(k\!+\!3) \!+\! \cdots \!+\!  k_{h-1} y(k\!+\!h)    \\
    &\cdots    \\
     &y(k\!+\!s\!-\!1) = k_1 y(k\!+\!s) + \cdots  +  k_{h-1} y(k\!+\!h\!+\!s\!-\!2). 
\end{aligned}
\end{small}
\end{equation}
Note that in \eqref{Hsigular}, the formula in the next line can be directly substituted into the previous line. After $(s-1)$ times of substitution, \eqref{Hsigular} is equivalently transformed into
\begin{equation}\label{Hifsingular}
\begin{small}
\begin{aligned}
    &y(k) \!=\!  f_1 y(k\!+\!s) \!+\! f_2 y(k\!+\!h\!+\!1) + \cdots +  f_{h-1} y(k\!+\!h\!+\!s\!-\!2),
\end{aligned}
\end{small}
\end{equation}
where $f_1,f_2,\cdots,f_{h-1}$, are constants related only to $k_1,k_2,\cdots,k_{h-1}$. It can be observed that the random variables in \eqref{Hifsingular} are $k_1,k_2,\cdots,k_{h-1}$  and $y(k+s),\cdots, y(k+h+s-2)$, a total of $(2h-2)$ random variables.

Considering that $s \geqslant h $, the linear space where $\mathcal H_{y}(k;h;s)$ is a singular matrix is a zero test set relative to the matrix space formed by all $\mathcal H_{y}(k;h;s)$. 
Thus, Lemma 1 is proved.

\end{proof}

\subsection{Proof of Theorem \ref{theMarkov}}
\begin{proof}
First, we derive the relationships between $x$ and $y,u$.
Expand $y(k)$ recursively until $y(k+h-1)$ by \eqref{sysmodel0}, we can obtain that
\begin{equation}\label{yocx}
Y(k;h) = O_c(h)x(k) + T(h) U(k;h).
\end{equation}
From Assumption 1, one infers that $\operatorname{rank}\left(O_c(h)\right)=m$ when $h \geqslant m$. Hence, there exists a matrix $ O^{\mathrm{L}}_c(h)$, s.t.,
\begin{equation}\label{ocl}
O^{\mathrm{L}}_c(h) O_c(h)=I.
\end{equation}
From \eqref{yocx} and \eqref{ocl}, we obtain 
\begin{equation}\label{xlr}
x(k)=O^{\mathrm{L}}_c(h) Y(k;h)-O^{\mathrm{L}}_c(h) T(h) U(k;h),
\end{equation}
i.e., $x(k)$ is a linear function of $Y(k;h)$ and $U(k;h)$.

Similarly, expand $x(k)$ recursively until $x(k+h-1)$ based on \eqref{sysmodel0}, we obtain that 
\begin{equation}\label{xhlr}
x(k+h)=A^{h} x(k) + O_b(h) U(k; h).
\end{equation}
It follows from \eqref{xlr} and \eqref{xhlr} that $x(k+h)$ is also a linear function of $ Y(k;h)$ and $U(k;h)$.
Thus, there exist constant matrices $F_1, F_2$, s.t.,
\begin{equation}\label{elx}
x(k+h)=F_1 Y(k ;h)+F_2 U(k;h).
\end{equation}
Eliminate $x$ by \eqref{elx} in \eqref{sysmodel0}, one follows that
\begin{equation}\nonumber
\begin{aligned}
y(k\!+\!h\!+\!t) = C A^{t}[F_1, F_2]\left[\!\begin{array}{c}
Y(k;h)\\
U(k;h)
\end{array}\!\right] + G(t) U(k\!+\!h;t).
\end{aligned}
\end{equation}
Denote $R=C A^{t}[F_1 \ F_2]$, the above equation is rewritten as 
\begin{equation}
y(k+h+t)=\left[R, G(t)\right]\left[\!\begin{array}{c}
Y(k;h) \\
U(k;h+t)
\end{array}\!\right].
\end{equation}
Considering the time from $(k\!+\!h\!+\!t)$ to $(k\!+\!h\!+\!t\!+\!s\!-\!1)$ for the output variable $y$, we obtain that
\begin{equation}\label{RG}
\textsf{Y}(k\!+\!h\!+\!t;s) =\left[R, G(t)\right] \mathcal L[y,u].
\end{equation}
By Lemma 1, it notes that $\mathcal L[y,u]$ is nonsingular with probability 1, i.e., $\mathcal L^{-1}[y,u]$ always exists.
Then, one infers from \eqref{RG} that 
\[\left[R, G(t)\right]=\textsf{Y}(k\!+\!h\!+\!t;s)\mathcal L^{-1}[y,u],\]
which implies that $G(t)$ satisfies \eqref{noisefreeidentify}.
Theorem \ref{theMarkov} is proved.
\end{proof}

\subsection{Proof of Theorem \ref{theTransformation}}
\begin{proof}
From \eqref{sysmodel0}, we obtain that
\begin{equation}\label{xkmv}
x(k\!+\!m\!-\!1) = A^{m\!-\!1}x(k) + O_b(m)U(k;m) + V(k;m).
\end{equation}
From \eqref{sysmodel} we have
\begin{equation}\label{xkme}
x(k\!+\!m\!-\!1) \!=\! A^{m\!-\!1}x(k) \!+\! O_b(m)U(k;m) \!+\! O_b(m)E(k;m).
\end{equation}
By Assumption 1, the system is controllable, i.e., 
\begin{equation}\nonumber
    \operatorname{rank} O_b(m) = m.
\end{equation}
Comparing \eqref{xkmv} and \eqref{xkme}, if we let
\begin{equation}\label{ek}
    E(k;m) =  O_b^{\mathrm{R}}(m) V(k;m),
\end{equation}
then $e$ in \eqref{sysmodel} is equivalent to $v$ in \eqref{sysmodel0}. 
Next, we prove that $e(k)$ is bounded and zero-mean.
From \eqref{ek}, we have
\begin{equation}\nonumber
\|E(k;m)\|_{\infty}  \leqslant   \left\|O_b^{\mathrm{R}}(m)\right\|_{\infty} \left\|V(k;m)\right\|_{\infty} .
\end{equation}
By Assumption 3, $v$ is bounded and zero-mean. Note that $A, O_b, m$ are constants or constant matrices, then $e(k)$ is bounded and zero-mean. Hence, Theorem \ref{theTransformation} is proved.
\end{proof}

\subsection{Proof of Theorem \ref{theorem5}}
\begin{proof}
Since $\mathcal{U}_{i+1} \subseteq \mathcal{U}_{i}$, $\mathcal{U}_{i+1} \neq \varnothing$ only if $\mathcal{U}_i \neq \varnothing$.

Then, we prove the sufficiency in Theorem \ref{theorem5}, i.e., $\mathcal{U}_{i+1} \neq \varnothing \ \text{if} \  \mathcal{U}_i \neq \varnothing$. 

Since $\mathcal{U}_i \neq \varnothing$, there exists a sequence of $U(k;h) \in \mathcal{U}_i$ such that $\|Y(k;h+i)\|_\infty \leqslant y_\mathrm{M}$.
It follows from the system model \eqref{sysmodel} that $\|Y(k;h+i)\|_\infty \leqslant y_\mathrm{M}$ is equivalent to
\begin{equation}
   \underline{x} \leqslant   X(k;h+i) \leqslant \overline{x},
\end{equation}
where $\underline{x}$ and $ \overline{x}$ are constants determined by $C, y_\mathrm{M}$ and the bound of noise $w$.


Expanding $x$ by the system model \eqref{sysmodel}, we have
\begin{equation}\label{xexpand}
\begin{aligned}
    x(k+h+i) =& A^{h+i} x(k) + O_b U(k+1;h-1+i) \\
    &+ A^{h-1}Bu(k) + T_e E(k;h+i),
\end{aligned}
\end{equation}
where $T_e$ is a constant matrix related to $A, B, h$ and the system order. 


Since the system is controllable, the matrix $O_b$ has full row rank when $h > m$, which means there exists a generalized right inverse matrix of $O_b$, denoted by $O_b^\mathrm{R}$, s.t.,
\begin{equation}\nonumber
O_b O_b^\mathrm{R}=I.
\end{equation}
Moreover, the vectors formed by the columns of $O_b^\mathrm{R}$ are linearly independent.
It follows that the existence of $U(k;h+i)$ is equivalent to the existence of $O_b U(k;h+i)$. 

By the existence of $U(k;h+i) \in \mathcal{U}_i$ such that $\|Y(k;h+i)\|_\infty \leqslant y_\mathrm{M}$ and \eqref{xexpand}, we have that there exists $O_b U(k;h+i-1)$ such that
\begin{equation}\label{xexist}
    \underline{x} \leqslant  A^{h+i-1}x(k)  + O_b U(k;h+i-1) \leqslant \overline{x}.
\end{equation}
Note that the influence of $E(k)$ has been considered in the constants $\underline{x}$ and $ \overline{x}$. 

Similarly, $\mathcal{U}_{i+1} \neq \varnothing$ is equivalent to the existence of $O_b U(k\!+\!1;h\!+\!i\!-\!1)$ such that
\begin{equation}\label{xprove}
   \underline{x} \leqslant A^{h\!+\!i} x(k) + O_b U(k\!+\!1;h\!+\!i\!-\!1) + A^{h\!+\!i\!-\!1}Bu(k) \leqslant \overline{x}
\end{equation}
Comparing \eqref{xexist} and \eqref{xprove}, we have the sufficient condition for \eqref{xprove}, i.e.,
\begin{equation}\label{compare1}
    \overline{x} - A^{h-1}x(k) \leqslant \min \left( \overline{x} - A^h x(k) - A^{h-1} B u(k) \right),
\end{equation} and
\begin{equation}\label{compare2}
        \underline{x} - A^{h-1}x(k) \geqslant \max \left( \underline{x} - A^h x(k) - A^{h-1} B u(k) \right).
\end{equation}
Note that $x(k\!+\!1) = Ax(k) + Bu(k) + v(k)$, \eqref{compare1} and \eqref{compare2} are equivalent to
\begin{equation}\label{theorem5fi}
  \min x(k\!+\!1) \leqslant x(k) - v(k) \leqslant \max x(k+1).
\end{equation}
By Assumption 1, the system is controllable, thus \eqref{theorem5fi} holds. Hence, we prove Theorem \ref{theorem5}.
\end{proof}

\subsection{Proof of Lemma 2}
\begin{proof}
Denote 
\begin{equation}\nonumber
\begin{aligned}
    D(u,\phi) = \left( F(y,u,\phi)u_2 + c(y,u,\phi) \right)^2,
\end{aligned}
\end{equation}
where $\phi = \{e, w\}$.

Denoting function \eqref{costu2} by $D_M(u)$,  we have 
\begin{equation}\nonumber
\begin{aligned}
    D_M(u) = D(u,\phi_0), 
\end{aligned}
\end{equation}
where $\phi_0 = \arg \max_{\phi}D(u,\phi)$.

$D_M(u)$ is convex related to $u$ if and only if
\begin{equation}\nonumber
\begin{aligned}
    D(\theta u_1 + (1-\theta) u_2, \phi_0) \leqslant & \theta D(u_1,\phi_1)  + (1-\theta) D(u_2,\phi_2),
\end{aligned}
\end{equation}
where $\theta \in [0,1]$ and  $\phi_0 = \arg \max_{\phi}D(\theta u_1 + (1-\theta) u_2, \phi)$, $\phi_1 = \arg \max_{\phi} D(u_1,\phi)$, $\phi_2 = \arg \max_{\phi} D(u_2,\phi)$.

Obviously, for all fixed $\phi$, $D(u,\phi)$ is a convex function related to $u$.
Hence, we have
\begin{equation}\nonumber
\begin{aligned}
    D(\theta u_1 + (1-\theta) u_2, \phi_0) \leqslant & \theta D(u_1,\phi_0)  + (1-\theta) D(u_2,\phi_0).
\end{aligned}
\end{equation}
Since $\phi_1 = \arg \max_{\phi} D(u_1,\phi)$, it follows that
\begin{equation}\nonumber
    D(u_1,\phi_0) \leqslant D(u_1,\phi_1).
\end{equation}
Similarly, we have 
\begin{equation}\nonumber
    D(u_2,\phi_0) \leqslant D(u_2,\phi_2).
\end{equation}
Therefore,
\begin{equation}\nonumber
\begin{aligned}
    & D(\theta u_1 + (1-\theta) u_2, \phi_0) \\ \leqslant & \theta D(u_1,\phi_0)  + (1-\theta) D(u_2,\phi_0)\\
    \leqslant & \theta D(u_1,\phi_1)  + (1-\theta) D(u_2,\phi_2).
\end{aligned}
\end{equation}
Hence, Lemma 2 is proved.

\end{proof}

\subsection{Proof of Theorem \ref{theorem7}}
\begin{proof}
Define $G_k$ as the estimation of $G$ under the $k$-th batch of the data $\{y_k, u_k\}$, i.e.,
\begin{equation}
G_k(t)= \textsf{Y}_k(d;s)\mathcal L^{-1}[y_k,u_k]\left[\!\begin{array}{l}
0 \\
I_{r}
\end{array}\!\right], \quad \forall k.
\end{equation}
Therefore,
\begin{equation}
G^{*}(t)=\frac{1}{N} \sum_{k=1}^{N} G_k^*(t), \quad
\hat{G}(t)=\frac{1}{N} \sum_{k=1}^{N} G_k(t). 
\end{equation}
Denote $\bm{w} = [w_1,w_2,\cdots,w_N]$, $\bm{\alpha} = [\alpha_1;\alpha_2;\cdots,\alpha_N]$, $\bm{y} = [y_1,y_2,\cdots,y_N]$ and $\bm{\beta} = [\beta_1;\beta_2;\cdots,\beta_N]$.
Let \[\Delta G=\left\|\hat{G}(t)-G^{*}(t)\right\|_{\mathrm{F}}^{2}\] 
Then, similar to \eqref{J}, we have
\begin{equation}\label{detG}
\begin{aligned}
\Delta G = & \frac{1}{N} \left\| \bm{w}\bm{\alpha} + \bm{ y \beta} \right\|_{\mathrm{F}}^{2} \\
 \leqslant & \frac{1}{N}\left(\left\| \bm{w}\bm{\alpha}  \right\|_{\mathrm{F}}^{2} + \left\| \bm{ y \beta} \right\|_{\mathrm{F}}^{2}\right).
\end{aligned}
\end{equation}
By Assumptions 2 and 3, $w_i,\alpha_i,\beta_i,y_i$ are independent of each other and bounded, then $\lim\limits_{N\to\infty} \Delta G= 0$.

When the noise obeys Gaussian distribution, by Assumption 3 and \eqref{paB}, $\beta$ obeys Gaussian distribution. 
According to the fact that $\alpha, \beta$ and $y$ are sequences with definite upper bounds and the Frobenius norm is Lipschitz continuous, \eqref{detG} is applicable to the case of Lemma A.1 in \cite{029}.
Hence,
\begin{equation}
\begin{aligned}
\Delta G \leqslant \frac{1}{N}\left( \delta^2\alpha_\mathrm{M}^2 + y_\mathrm{M}^2\alpha_\mathrm{M}^2 \right)(2\sqrt{n+s} + \tau)^2,
\end{aligned}
\end{equation}
with probability at least $1 - 2 exp(-\tau^2/2)$. Considering that $\alpha_\mathrm{M}$, $\delta$, $s$, $n$, $y_{\mathrm{M}}$ are all constants, the proof is completed.
\end{proof}

\subsection{Proof of Theorem \ref{theorem8}}
\begin{proof}
Define the maximum deviation at $N$ iteration, i.e., at the time of $N$ batches data are used in identification as $D_N$.
Define the maximum deviation of estimation of $G$ at $N$ iteration as $\delta D_N$.

It follows from  \eqref{costab} that
\begin{equation}\nonumber
\begin{aligned}
    \delta D_N^2 \leqslant \min _{u}  \sum_{j=r+1}^{s}&\left(\max _{e,w}\sum_{i=r+1}^{s}\beta(i,j)y(d+i-1) \right. 
    \\ &\left.+\max _{e,w}\sum_{i=r+1}^{s}\alpha(i,j)w(d+i-1) \right)^2.
\end{aligned}
\end{equation}
Since $\| \bm{w}\|_{\infty} \leqslant 2\delta$, we have
\begin{equation}\nonumber
\begin{aligned}
    \delta D_N^2 \leqslant  & \min _{u}  \sum_{j=r+1}^{s} \left(\max _{e,w}\sum_{i=r+1}^{s}\beta(i,j)y_\mathrm{M}
    +2\delta \sum_{i=r+1}^{s}\alpha(i,j)\right)^2\\
    \leqslant & \min _{u}  \sum_{j=r+1}^{s} \left(\sum_{i=r+1}^{s}\beta(i,j)_{\max} y_\mathrm{M}
    +2\delta \sum_{i=r+1}^{s}\alpha(i,j)\right)^2.
\end{aligned}
\end{equation}

Then, we analyze the range of $\beta(i,j)$.
Since
\begin{equation}\nonumber
    \beta= \mathcal L^{-1}[y^*\!+\!w_i,u^*\!+\!e_i] - \mathcal L^{-1}[y^*\!+\!w_j,u^*\!+\!e_j],
\end{equation}
and by Assumption 3, $e,w$ is relatively small. Hence, we have
\begin{equation}\nonumber
    \beta = \mathcal L^{-1}[y^*,u^*] (e_i-e_j,w_i-w_j) \mathcal L^{-1}[y^*,u^*].
\end{equation}
Hence, 
\begin{equation}\nonumber
    \beta(i,j)_{\max} \leqslant 2\delta \alpha_{\mathrm{M}}^2.
\end{equation}

Therefore, we have
\begin{equation}\label{deltaD}
\begin{aligned}
   \delta D_N^2 \leqslant & \sum_{j=r+1}^{s} \left(\sum_{i=r+1}^{s} 2\delta \alpha_{\mathrm{M}}^2 y_\mathrm{M}
    +2\delta \sum_{i=r+1}^{s}\alpha_{\mathrm{M}}\right)^2 \\
    = & 4\delta^2(s-r)^3\alpha_{\mathrm{M}}^2(1+\alpha_{\mathrm{M}}y_\mathrm{M})^2
\end{aligned}
\end{equation}

From the definition of $D_N$, we have
\begin{equation}\nonumber
\begin{aligned}
    D_N =\| \frac{1}{N}(\sum_{i = 1}^{N-1}G_{i} + G_{N_1}) - \frac{1}{N}(\sum_{i = 1}^{N-1}G_{i} + G_{N_2})\|_\mathrm{F},
\end{aligned}
\end{equation}
where $G_{i}, i=1,2,\cdots,N-1$ are the identification result at $i$ iteration, $G_{N_1}$ and $G_{N_2}$ are the two identification results of $G$ with the largest difference at $N$ iteration.
Hence, 
\begin{equation}
\begin{aligned}
    D_N =\frac{1}{N}\delta D_N\leqslant \frac{1}{N} 2\delta(s-r)^{\frac{3}{2}}\alpha_{\mathrm{M}}(1+\alpha_{\mathrm{M}}y_\mathrm{M}).
\end{aligned}
\end{equation}
Since $\alpha_{\mathrm{M}}, y_\mathrm{M}, \delta, s, r$ are constants, Theorem \ref{theorem8} is proved.
\end{proof}

For the case where white noise is used as the system identification input, we only have the probability convergence of the maximum identification deviation. Furthermore, the the convergence speed is slow compared to the designed $u$ in this paper. We analyze the case of white noise input as follows.

Define the maximum deviation at $N$ iteration of identification using white noise input as $\tilde{D}_N$ and the maximum deviation of estimation of $G$ using white noise input at $N$ iteration as $\delta \tilde{D}_N$.
Similar to \eqref{deltaD}, we have
\begin{equation}\label{deltawhitenoiseD}
\begin{aligned}
   \delta \tilde{D}_N^2 \leqslant & \sum_{j=r+1}^{s} \left(\sum_{i=r+1}^{s} 2\delta \alpha(i,j)^2 y_\mathrm{M}
    +2\delta \sum_{i=r+1}^{s}\alpha(i,j)\right)^2. 
\end{aligned}
\end{equation}

However, since the range of $\alpha(i,j)$ cannot be determined, we cannot obtain the maximum value of $\delta \tilde{D}_N$.

Since $\alpha(i,j)$ at $N$ iteration is determined by the input signal $u$ at $N$ iteration.
Define the change of $u$ as $\Delta u$. By the definition of the identification function of $G$, we obtain that $u$ is in the $s$-th row and $s$-th column of the matrix  $\mathcal L[y,u]$. 
Define the change of $\alpha(i,j)$ as $\Delta \alpha(i,j)$ when $u$ changes. We have the following equations
\begin{equation}\label{changeofalpha}
    \Delta \alpha(i,j) = -\frac{\alpha(i,s)\alpha(s,j)\Delta u}{1+\alpha(s,s)\Delta u}=-\frac{\alpha(i,s)\alpha(s,j)}{\frac{1}{\Delta u}+\alpha(s,s)}.
\end{equation}
Since $u$ is a white noise input, $\Delta u$ is also a white noise input. Although $u$ and $\Delta u$ are bounded, we cannot ensure that $\Delta \alpha(i,j)$ is bounded, i.e., when $(\frac{1}{\Delta u}+\alpha(s,s))$ is close to $0$, there is no maximum value for $\alpha(i,j)$.

Considering that $u$ obeys a distribution, $\Delta \alpha(i,j)$ also obeys a distribution. Therefore, we can get the conclusion that the bound of $\Delta \alpha(i,j)$ probability converges according to the distribution of $u$.

For example, when $\Delta u$ obeys a distribution which makes $\Delta \alpha(i,j)^2$ obeys Gaussian distribution. By  Lemma A.1 in \cite{029}, we have the following convergence rate at a probability at least $1-2e^{-t^2/2}$.
\begin{equation}
\begin{aligned}
    D_N = & \frac{1}{N}\delta D_N = \|[\frac{1}{N}, \cdots, \frac{1}{N}][\tilde{D}_N,\tilde{D}_N,\cdots,\tilde{D}_N]^\mathrm{T}\|_\mathrm{F}\\ \leqslant & \frac{1}{N}\sqrt{2(N+1)+t}.
\end{aligned}
\end{equation}
In this example, the convergence rate is $\mathcal{O}(1/\sqrt{N})$.

\balance
\bibliographystyle{IEEEtran}

\bibliography{ref} 

\begin{thebibliography}{10}
\providecommand{\url}[1]{#1}
\csname url@samestyle\endcsname
\providecommand{\newblock}{\relax}
\providecommand{\bibinfo}[2]{#2}
\providecommand{\BIBentrySTDinterwordspacing}{\spaceskip=0pt\relax}
\providecommand{\BIBentryALTinterwordstretchfactor}{4}
\providecommand{\BIBentryALTinterwordspacing}{\spaceskip=\fontdimen2\font plus
\BIBentryALTinterwordstretchfactor\fontdimen3\font minus
  \fontdimen4\font\relax}
\providecommand{\BIBforeignlanguage}[2]{{%
\expandafter\ifx\csname l@#1\endcsname\relax
\typeout{** WARNING: IEEEtran.bst: No hyphenation pattern has been}%
\typeout{** loaded for the language `#1'. Using the pattern for}%
\typeout{** the default language instead.}%
\else
\language=\csname l@#1\endcsname
\fi
#2}}
\providecommand{\BIBdecl}{\relax}
\BIBdecl

\bibitem{previ}
X.~Mao, J.~He, and C.~Zhao, ``An improved subspace identification method with
  variance minimization and input design,'' in \emph{submitted to IEEE ACC},
  2022.

\bibitem{040}
S.~J. Qin, ``An overview of subspace identification,'' \emph{Computers \&
  Chemical Engineering}, vol.~30, no.~10, pp. 1502--1513, 2006.

\bibitem{LJUNG20101}
L.~Ljung, ``Perspectives on system identification,'' \emph{Annual Reviews in
  Control}, vol.~34, no.~1, pp. 1--12, 2010.

\bibitem{7394584}
N.~S. Özbek and I.~Eker, ``A novel interactive system identification and
  control toolbox dedicated to real-time identification and model reference
  adaptive control experiments,'' in \emph{2015 9th International Conference on
  Electrical and Electronics Engineering (ELECO)}, 2015, pp. 859--863.

\bibitem{8951600}
T.~Phillips, H.~Mehrpouyan, J.~Gardner, and S.~Reese, ``A covert system
  identification attack on constant setpoint control systems,'' in \emph{2019
  Seventh International Symposium on Computing and Networking Workshops
  (CANDARW)}, 2019, pp. 367--373.

\bibitem{7057677}
T.~R. Nudell, S.~Nabavi, and A.~Chakrabortty, ``A real-time attack localization
  algorithm for large power system networks using graph-theoretic techniques,''
  \emph{IEEE Transactions on Smart Grid}, vol.~6, no.~5, pp. 2551--2559, 2015.

\bibitem{8897147}
J.~Schoukens and L.~Ljung, ``Nonlinear system identification: A user-oriented
  road map,'' \emph{IEEE Control Systems Magazine}, vol.~39, no.~6, pp. 28--99,
  2019.

\bibitem{8357807}
L.~Lu, H.~Zhao, and B.~Champagne, ``Distributed nonlinear system identification
  in <inline-formula><tex-math notation="latex">$\alpha$
  </tex-math></inline-formula>-stable noise,'' \emph{IEEE Signal Processing
  Letters}, vol.~25, no.~7, pp. 979--983, 2018.

\bibitem{ASTROM1971123}
K.~Åström and P.~Eykhoff, ``System identification—a survey,''
  \emph{Automatica}, vol.~7, no.~2, pp. 123--162, 1971.

\bibitem{030}
Z.~Zhang, R.~Du, and R.~V. Cowlagi, ``Randomized sampling-based trajectory
  optimization for uavs to satisfy linear temporal logic specifications,''
  \emph{Aerospace Science and Technology}, vol.~96, no. 1–4, p. 105591, 2020.

\bibitem{032}
K.~Astrom, ``Maximum likelihood and prediction error methods,'' \emph{IFAC
  Proceedings Volumes}, vol.~12, no. 8, Supplement 1, pp. 551--574, 1979.

\bibitem{002}
C.~{Yu}, L.~{Ljung}, A.~{Wills}, and M.~{Verhaegen}, ``Constrained subspace
  method for the identification of structured state-space models (cosmos),''
  \emph{IEEE Transactions on Automatic Control}, vol.~65, no.~10, pp.
  4201--4214, 2020.

\bibitem{036}
A.~Micchi and G.~Pannocchia, ``Comparison of input signals in subspace
  identification of multivariable ill-conditioned systems,'' \emph{Journal of
  Process Control}, vol.~18, no.~6, pp. 582--593, 2008.

\bibitem{034}
M.~Verhaegen and V.~Verdult, \emph{Filtering and System Identification: A Least
  Squares Approach}, 1st~ed.\hskip 1em plus 0.5em minus 0.4em\relax USA:
  Cambridge University Press, 2007.

\bibitem{031}
N.~Everitt, M.~Galrinho, and H.~Hjalmarsson, ``Open-loop asymptotically
  efficient model reduction with the steiglitz–mcbride method,''
  \emph{Automatica}, vol.~89, pp. 221--234, 2018.

\bibitem{004}
C.~Yu, J.~Chen, and M.~Verhaegen, ``Subspace identification of individual
  systems in a large-scale heterogeneous network,'' \emph{Automatica}, vol.
  109, p. 108517, 2019.

\bibitem{021}
M.~Verhaegen and A.~Hansson, ``N2sid: Nuclear norm subspace identification of
  innovation models,'' \emph{Automatica}, vol.~72, pp. 57--63, 2016.

\bibitem{007}
S.~K. Perepu and A.~K. Tangirala, ``Identification of equation error models
  from small samples using compressed sensing techniques,''
  \emph{IFAC-PapersOnLine}, vol.~48, no.~8, pp. 795--800, 2015.

\bibitem{029}
S.~{Oymak} and N.~{Ozay}, ``Non-asymptotic identification of lti systems from a
  single trajectory,'' in \emph{IEEE ACC}, 2019, pp. 5655--5661.

\bibitem{001}
C.~Yu, L.~Ljung, and M.~Verhaegen, ``Identification of structured state-space
  models,'' \emph{Automatica}, vol.~90, pp. 54--61, 2018.

\bibitem{006}
F.~Z. Chaoui, F.~G. *, Y.~Rochdi, M.~Haloua, and A.~Naitali, ``System
  identification based on hammerstein model,'' \emph{International Journal of
  Control}, vol.~78, no.~6, pp. 430--442, 2005.

\bibitem{005}
W.~Zhao, G.~Yin, and E.-W. Bai, ``Sparse system identification for stochastic
  systems with general observation sequences,'' \emph{Automatica}, vol. 121, p.
  109162, 2020.

\bibitem{013}
K.~Peternell, W.~Scherrer, and M.~Deistler, ``Statistical analysis of novel
  subspace identification methods,'' \emph{Signal Processing}, vol.~52, no.~2,
  pp. 161--177, 1996.

\bibitem{003}
C.~{Yu} and M.~{Verhaegen}, ``Subspace identification of distributed clusters
  of homogeneous systems,'' \emph{IEEE Transactions on Automatic Control},
  vol.~62, no.~1, pp. 463--468, 2017.

\bibitem{016}
A.~{Haber} and M.~{Verhaegen}, ``Subspace identification of large-scale
  interconnected systems,'' \emph{IEEE Transactions on Automatic Control},
  vol.~59, no.~10, pp. 2754--2759, 2014.

\bibitem{012}
J.~S. Grover, C.~Liu, and K.~Sycara, ``Parameter identification for multirobot
  systems using optimization based controllers (extended version),'' 2020.

\bibitem{015}
I.~Hajizadeh, M.~Rashid, K.~Turksoy, S.~Samadi, J.~Feng, M.~Sevil, N.~Frantz,
  C.~Lazaro, Z.~Maloney, E.~Littlejohn, and A.~Cinar, ``Multivariable recursive
  subspace identification with application to artificial pancreas systems,''
  \emph{IFAC-PapersOnLine}, vol.~50, no.~1, pp. 886--891, 2017.

\bibitem{017}
M.~Inoue, ``Subspace identification with moment matching,'' \emph{Automatica},
  vol.~99, pp. 22--32, 2019.

\bibitem{025}
B.~{Wahlberg}, H.~{Hjalmarsson}, and M.~{Annergren}, ``On optimal input design
  in system identification for control,'' in \emph{IEEE CDC}, 2010, pp.
  5548--5553.

\bibitem{024}
M.~{Annergren}, C.~A. {Larsson}, H.~{Hjalmarsson}, X.~{Bombois}, and
  B.~{Wahlberg}, ``Application-oriented input design in system identification:
  Optimal input design for control [applications of control],'' \emph{IEEE
  Control Systems Magazine}, vol.~37, no.~2, pp. 31--56, 2017.

\bibitem{041}
M.~Gevers, L.~Mišković, D.~Bonvin, and A.~Karimi, ``Identification of
  multi-input systems: variance analysis and input design issues,''
  \emph{Automatica}, vol.~42, no.~4, pp. 559--572, 2006.

\bibitem{023}
K.~Lindqvist and H.~Hjalmarsson, ``Identification for control: adaptive input
  design using convex optimization,'' in \emph{IEEE CDC}, 2001, pp. 4326--4331.

\bibitem{027}
M.~{Casini}, A.~{Garulli}, and A.~{Vicino}, ``Input design in worst-case system
  identification using binary sensors,'' \emph{IEEE Transactions on Automatic
  Control}, vol.~56, no.~5, pp. 1186--1191, 2011.

\bibitem{026}
I.~R. {Manchester}, ``Input design for system identification via convex
  relaxation,'' in \emph{IEEE CDC}, 2010, pp. 2041--2046.

\bibitem{STOJANOVIC2016576}
V.~Stojanovic and N.~Nedic, ``Robust identification of oe model with
  constrained output using optimal input design,'' \emph{Journal of the
  Franklin Institute}, vol. 353, no.~2, pp. 576--593, 2016.

\bibitem{MU2018327}
B.~Mu and T.~Chen, ``On input design for regularized lti system identification:
  Power-constrained input,'' \emph{Automatica}, vol.~97, pp. 327--338, 2018.

\bibitem{FUJIMOTO201837}
Y.~Fujimoto and T.~Sugie, ``Informative input design for kernel-based system
  identification,'' \emph{Automatica}, vol.~89, pp. 37--43, 2018.

\bibitem{022}
Y.~{Zheng} and N.~{Li}, ``Non-asymptotic identification of linear dynamical
  systems using multiple trajectories,'' \emph{IEEE Control Systems Letters},
  vol.~5, no.~5, pp. 1693--1698, 2021.

\bibitem{8814438}
S.~Oymak and N.~Ozay, ``Non-asymptotic identification of lti systems from a
  single trajectory,'' in \emph{2019 American Control Conference (ACC)}, 2019,
  pp. 5655--5661.

\bibitem{043}
B.~L. HO and R.~E. Kalman, ``Editorial: Effective construction of linear
  state-variable models from input/output functions,'' \emph{at -
  Automatisierungstechnik}, vol.~14, no. 1-12, pp. 545--548, 1966.

\bibitem{038}
N.~I.~M. Gould and P.~L. Toint, \emph{Numerical Methods for Large-Scale
  Non-Convex Quadratic Programming}.\hskip 1em plus 0.5em minus 0.4em\relax
  Boston, MA: Springer US, 2002, pp. 149--179.

\end{thebibliography}

\begin{IEEEbiographynophoto}{Xiangyu Mao}
(S'21) received the B.E. degree in Department of Automation from Tsinghua University, Beijing, China, in 2020. 
He is currently working toward the Ph.D. degree with the Department of Automation, Shanghai Jiaotong University, Shanghai, China. 
He is a member of Intelligent of Wireless Networking and Cooperative Control group. 
His research interests include system identification, networked systems and  distributed  optimization in multi-agent networks. 
\end{IEEEbiographynophoto}

\begin{IEEEbiographynophoto}{Jianping He} 
(SM'19) is currently an associate professor in the Department of Automation at Shanghai Jiao Tong University. He received the Ph.D. degree in control science and engineering from Zhejiang University, Hangzhou, China, in 2013, and had been a research fellow in the Department of Electrical and Computer Engineering at University of Victoria, Canada, from Dec. 2013 to Mar. 2017. His research interests mainly include the distributed learning, control and optimization, security and privacy in network systems.

Dr. He serves as an Associate Editor for IEEE Open Journal of Vehicular Technology and KSII Trans. Internet and Information Systems. He was also a Guest Editor of IEEE TAC, IEEE TII, International Journal of Robust and Nonlinear Control, etc. He was the winner of Outstanding Thesis Award, Chinese Association of Automation, 2015. He received the best paper award from IEEE WCSP'17, the best conference paper award from IEEE PESGM'17, the finalist best student paper award from IEEE ICCA'17, and the finalist best conference paper award from IEEE VTC'20-Fall.
\end{IEEEbiographynophoto}

\begin{IEEEbiographynophoto}{Chengcheng Zhao}
 (M'18) received her PhD degree in control science and engineering from Zhejiang University, Hangzhou, China, in 2018. She is currently an associate researcher in the College of Control Science and Engineering, Zhejiang University. She worked as a postdoctoral fellow in the College of Control Science and Engineering, Zhejiang University from 2018 to 2021, and worked as a postdoctoral fellow at the ECE department, University of Victoria, from 2019 to 2020. Her research interests include consensus and distributed optimization, distributed energy management in smart grids, vehicle platoon, and security and privacy in network systems. She received IEEE PESGM 2017 best conference papers award, and one of her papers was shortlisted in IEEE ICCA 2017 best student paper award finalist. She serves as an editor of Wireless Networks since 2021.
\end{IEEEbiographynophoto}

\end{document}